\documentclass[10pt,conference]{IEEEtran}

\usepackage[utf8]{inputenc}
\usepackage[colorlinks,bookmarksopen,bookmarksnumbered,citecolor=blue,urlcolor=black,linkcolor=blue]{hyperref}

\usepackage{pifont}  
\newcommand{\cmark}{\ding{51}} 
\newcommand{\xmark}{\ding{55}} 

\usepackage{mathtools}
\usepackage{amsthm}
\usepackage{amssymb}
\usepackage{amsmath}
\usepackage{yhmath}
\usepackage{amsfonts}
\usepackage{url}
\usepackage{listings}
\usepackage{color}
\usepackage{xcolor}
\usepackage{graphicx}
\usepackage{caption}
\usepackage{subcaption}
\usepackage[noadjust]{cite}
\usepackage{breqn}
\usepackage{makecell}
\usepackage{enumitem}
\usepackage[ruled,linesnumbered]{algorithm2e}

\usepackage{array}
\usepackage{balance}
\usepackage{amssymb}
\usepackage{multirow}
\usepackage{multicol}
\usepackage{bm}
 \usepackage{tikz}

\usepackage{algpseudocode}
\usepackage{times}
\usepackage{fancyhdr}

\include{pythonlisting}

\newtheorem{definition}{Definition}

\newtheorem{lemma}{Lemma}
\newtheorem{limitation}{Limitation}
\newtheorem{corollary}{Corollary}
\newtheorem{example}{Example}
\newtheorem{ruledef}{Rule}

\definecolor{darkolivegreen}{rgb}{0.13, 0.55, 0.13}

\newcommand{\name}[1]{LEFT-RS}

\IEEEoverridecommandlockouts

\begin{document}



\bstctlcite{IEEEexample:BSTcontrol}

\title{LEFT-RS: A Lock-Free Fault-Tolerant Resource Sharing Protocol for Multicore Real-Time Systems\thanks{Corresponding author: Shuai Zhao, zhaosh56@mail.sysu.edu.cn}}


 \author{
 
 \IEEEauthorblockN{Nan Chen\IEEEauthorrefmark{1}, Xiaotian Dai\IEEEauthorrefmark{1}, Tong Cheng\IEEEauthorrefmark{2}, Alan Burns\IEEEauthorrefmark{1}, Iain Bate\IEEEauthorrefmark{1}, Shuai Zhao\IEEEauthorrefmark{2}}
 
  \IEEEauthorblockA{
 \IEEEauthorrefmark{1}University of York, UK
 \IEEEauthorrefmark{2}Sun Yat-sen University, China
}
 }

\maketitle

    

\begin{abstract}
Emerging real-time applications have driven the transition to multicore embedded systems, where tasks must share resources due to functional demands and limited availability. These resources, whether local or global, are protected within critical sections to prevent race conditions, with locking protocols ensuring both exclusive access and timing requirements. However, transient faults occurring within critical sections can disrupt execution and propagate errors across multiple tasks. Conventional locking protocols fail to address such faults, and integrating traditional fault tolerance techniques often increases blocking. Recent approaches improve fault recovery through parallel replica execution; however, challenges remain due to sequential accessing, coordination overhead, and susceptibility to common-mode faults. 
In this paper, we propose a Lock-frEe Fault-Tolerant Resource Sharing (LEFT-RS) protocol for multicore real-time systems. LEFT-RS allows tasks to concurrently access and read global resources while entering their critical sections in parallel.
Each task can complete its access earlier upon successful execution if other tasks experience faults, thereby improving the efficiency of resource usage. Our design also limits the overhead and enhances fault resilience.
We present a comprehensive worst-case response time analysis to ensure timing guarantees. Extensive evaluation results demonstrate that our method significantly outperforms existing approaches, achieving up to an $84.5\%$ improvement in schedulability on average.
\end{abstract}


\section{Introduction}
\label{sec:intro}


The growing demands of emerging real-time applications have driven the shift from single-core to multicore embedded systems, where tasks frequently access shared resources due to functional needs and inherent resource constraints. These shared resources include both local and global resources. Local shared resources reside within a single core but may be accessed by multiple threads on that core, such as per-core queues or synchronisation flags, whereas global shared resources span multiple cores and include memory-mapped I/O devices and system-wide data structures. To avoid race conditions, such resources are typically protected within critical sections, which are code segments that require mutual exclusion. Resource-locking protocols are designed to guarantee mutually exclusive access for tasks while maintaining the timing requirements of the systems~\cite{brandenburg2022multiprocessor}. 


Embedded systems in critical domains must continue to function correctly in the presence of faults.
Transient faults, which are temporary errors that occur during execution due to hardware fluctuations, timing anomalies, or other short-lived conditions, are of particular concern \cite{vijaykumar2002transient}. They are often tolerated by temporal redundancy (e.g., re-execution or checkpointing, where checkpointing enables small-scope re-execution), or spatial redundancy (e.g., replication) \cite{osinski2017survey}.
Transient faults during critical sections must be handled to prevent error propagation across tasks. 
Conventional resource-locking protocols do not explicitly address fault tolerance \cite{brandenburg2022multiprocessor}.
Directly integrating checkpointing into resource-sharing protocols can significantly exacerbate global resource contention \cite{chen2022msrp}, primarily because fault recovery may substantially prolong the execution of critical sections, leading to locks being held for extended periods and thus increasing blocking times for other tasks in the system.

To address this challenge, MSRP-FT \cite{chen2022msrp} combines checkpointing and replication, allowing spinning tasks to execute replicas in parallel and thereby reducing time lost to transient faults during critical-section execution. However, resource requests are still served sequentially in FIFO order, which can cause prolonged blocking under frequent faults. Moreover, coordinating replica execution introduces additional overhead, and the reliance on identical replicas makes the system vulnerable to correlated or timing-sensitive faults. These limitations are examined in detail in Section~\ref{sec:limitation}.

In this paper, we propose a Lock-frEe Fault-Tolerant Resource Sharing protocol for multicore real-time systems, referred to as LEFT-RS. 
We record resource request orders in a FIFO queue. 
Instead of managing resource access sequentially as in MSRP-FT, LEFT-RS adopts a lock-free design that enables tasks to concurrently read global resources and initiate local execution. Tasks re-execute their critical sections either when they incur faults or when the shared resource has been updated, to ensure fault-free execution and avoid race conditions.  
Each task can complete its access upon successful execution when preceding tasks in the FIFO queue encounter faults. 
LEFT-RS facilitates resource access efficiency and ensures that remote blocking (blocking caused by tasks on remote cores due to resource contention) does not worsen as fault occurrences increase. Moreover, its design avoids the heavy coordination overhead and enhances fault resilience. To guarantee timing predictability, we provide a comprehensive worst-case response time analysis. Extensive experimental evaluations demonstrate that our method outperforms the state-of-the-art by up to $84.5\%$ on average in terms of schedulability.

The following sections are structured as follows: Section~\ref{sec:model} presents the system model. Section~\ref{sec:related} reviews related work and provides an in-depth analysis of the limitations of state-of-the-art approaches. Section~\ref{sec:protocol} introduces the proposed protocol, while Section~\ref{sec:performance} offers a theoretical comparison with the state-of-the-art approach.
Section~\ref{sec:analysis} presents the worst-case response time analysis.
Section~\ref{sec:evaluation} shows the evaluation results 
and Section~\ref{sec:conclusion} concludes the paper.

\section{System Model} 
\label{sec:model}

We consider systems composed of a set of identical cores, denoted as $\Lambda$, and a set of sporadic tasks, $\Gamma$, scheduled using the FP-FPS (Fully-Partitioned Fixed-Priority Scheduling) scheme. Each core in $\Lambda$ is represented as $\lambda_k$, where $k$ identifies the index of the core.
Each task $\tau_i$ (the $i$\textsuperscript{th} task in $\Gamma$) is characterized by the tuple $\tau_i = \{C_i, T_i, D_i, P_i\}$. Here, $C_i$ represents the pure Worst-Case Execution Time (WCET) without accessing shared resources, $T_i$ denotes the period (or minimum inter-arrival time), $D_i$ is the constrained deadline satisfying $D_i \leq T_i$, and $P_i$ indicates the priority of the task. Each task is assigned a unique priority, with higher values of $P_i$ corresponding to higher priorities. For a given task $\tau_i$, we denote tasks on the same core with higher and lower priorities as $\tau_h$ and $\tau_l$, respectively. Tasks assigned to other cores are denoted by $\tau_j$.

The system also includes a set of shared resources, denoted as $\mathcal{R}$, where the $x$\textsuperscript{th} shared resource is represented as $r^x$. 
Each resource $r^x$ is characterised by two parameters: $c^x$ and $N^x_i$. The parameter $c^x$ represents the computation length associated with $r^x$, while $N^x_i$ indicates the number of requests made by task $\tau_i$ to  $r^x$ during a single release.  
This work does not consider nested resource sharing, meaning that a task can hold only one resource at a time. However, group locks~\cite{zhao2017new} can be directly supported in the system.

We address transient faults in this study, which can be mitigated through redundancy techniques such as checkpointing or replication. Each transient fault is assumed to affect only one task at a time. Fault detection is performed at designated checkpoints using an acceptance test, such as result validation or consistency checking, to verify the correctness of task outputs \cite{izosimov2005design}.
The frequency of occurrence of transient faults can be modelled through various ways~\cite{miskov2008process}, which is out of the scope of this paper. Similar to~\cite {chen2022msrp}, we assume each task can incur a maximum of $f_i$ faults during any single release. 
The number of faults a task can incur during its access to $r^x$ is denoted as  $f^x_i \leq f_i$. 
We use $n^x_i$ to denote the total execution number of a request of $\tau_i$ to $r^x$, which contains $f^x_i$ failures and one successful execution ($n^x_i=f^x_i+1$). 
All the notations are summarised in Table \ref{tab:notation}.


\begin{table}[t]
    \centering
    \renewcommand{\arraystretch}{1.1}
    \setlength{\tabcolsep}{8pt}
    \caption{List of Notations Used in the System Model}
    \begin{tabular}{l  p{5.8cm}}
        \hline
        \textbf{Notations} & \textbf{Descriptions} \\
        \hline
        $\Lambda$ & Set of identical cores in the system \\
        $\lambda_k$ & Core with index $k$, $\lambda_k \in \Lambda$ \\
        \hline
        $\Gamma$ & Set of sporadic tasks \\
        $\tau_i$ & Task with index $i$, $\tau_i \in \Gamma$ \\
        $\tau_h$, $\tau_l$, $\tau_j$ & Tasks related to $\tau_i$: higher/lower priority on the same core, or on a different core \\
        \hline
        $C_i$ & Pure worst-case execution time (WCET) of   $\tau_i$ \\
        $T_i$, $D_i$ & Period (minimum inter-arrival time) and constrained deadline ($D_i \leq T_i$) of   $\tau_i$  \\
        $P_i$ & Priority of $\tau_i$ (a larger $P_i$ means a higher priority) \\
        \hline

        $\mathcal{R}$ & Set of shared resources in the system\\
        $r^x$ & Shared resource with index $x$ \\
        $c^x$ & Computation length of $r^x$ \\
        $N^x_i$ & Number of requests for $r^x$ by $\tau_i$ per release \\
        $f_i$, $f^x_i$, $n^x_i$ & Number of faults incurred by $\tau_i$ per release. Fault and total execution number during access to $r^x$ \\
        \hline
    \end{tabular}
    \label{tab:notation}
      \vspace{-1em}
\end{table}




\section{Related Work and Background} 
\label{sec:related}

In this section, we first introduce the literature on resource sharing protocols and then review the state-of-the-art in fault-aware resource sharing.

\subsection{Resource sharing protocols}\label{sec:related_resource}



Lock-based protocols have been widely studied for efficient resource sharing in multicore systems. The Multiprocessor Priority Ceiling Protocol (MPCP) extends traditional priority ceilings to reduce remote blocking \cite{gai2003comparison}. Multiprocessor Stack Resource Protocol (MSRP) uses non-preemptive spinning to avoid preemption delay \cite{gai2003comparison}. In contrast, the Multiprocessor Resource Sharing Protocol (MrsP) employs a global priority ceiling to bound blocking, introduces a helping mechanism, and allows cross-core migration \cite{burns2013schedulability}. The Flexible Multiprocessor Locking Protocol (FMLP) distinguishes between short and long critical sections, combining spinning and suspension to balance blocking and concurrency \cite{block2007flexible}. The ${O}(m)$ Locking Protocol (OMLP) further refines priority-based mechanisms for improved scalability \cite{brandenburg2013omlp}, and the Flexible Resource Accessing Protocol (FRAP) introduces dynamic spinning priorities and advanced blocking analysis for better performance~\cite{zhao2024frap}. Each protocol targets specific challenges, and their effectiveness depends on system architecture and application requirements.

In contrast, lock-free approaches \cite{anderson1997real} rely on atomic operations, such as Compare-and-Swap (CAS) or Load-Link/Store-Conditional (LL/SC), to update shared resources without using explicit locks. Instead of enforcing mutual exclusion through serialized access, these algorithms detect conflicts by monitoring changes to shared data and retry operations when necessary. This non-blocking design ensures system-wide progress while preventing race conditions and data corruption.

\subsection{Fault-aware resource sharing}\label{sec:MSRP-FT}


The work in~\cite{nabavi2023fault} addresses transient faults in critical sections but adopts simple checkpointing, which tolerates faults through repeated sequential re-execution of resource requests.
Building upon the traditional MSRP~\cite{gai2001minimizing}, MSRP-FT \cite{chen2022msrp} efficiently tolerates transient faults within the critical sections of tasks in Mixed-Criticality Systems (MCS). Since the fault tolerance approach of MSRP-FT is independent of MCS, we will discuss this approach within a standard system setup to ease the presentation. 

\begin{figure}[h]
\centering
\includegraphics[width=0.9\columnwidth]{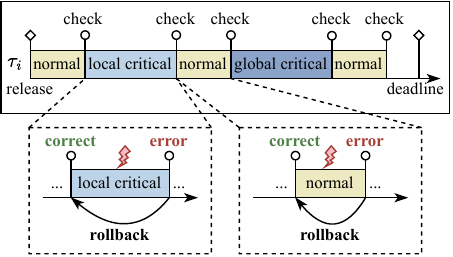}
\caption{Fault tolerance of normal and local critical sections}
\label{fig:check1}
  \vspace{-1em}
\end{figure}

As shown in Figure \ref{fig:check1}, MSRP-FT establishes checkpoints not only at the beginning and end of each task but also around critical sections within a task’s execution. 
These checkpoints divide execution into normal sections (no shared resources) and critical sections, where shared resources are accessed. Critical sections include local (intra-core) and global (inter-core) resource access. Each checkpoint detects faults for the previous segment and saves the states for the next segment (if applicable).
The setup enables faults to be tolerated in time after the execution of a critical section, preventing transitive errors from affecting other tasks and avoiding unnecessary large-scale re-executions.

If a fault is detected in a normal section, the system rolls back to the most recent checkpoint and re-executes the faulty segment only.
When managing local shared resources, the same ceiling protocol is adopted from MSRP.  Each local resource $r^x$ is assigned a ceiling priority equal to the highest priority of any task that accesses it. When a task locks a local resource, its priority is temporarily elevated to the resource’s ceiling priority. As shown in Figure \ref{fig:check1}, the fault-tolerance approach applied is the same as a normal section where the system initiates rollback and re-execution.

\begin{figure}[h]
\centering
\includegraphics[width=0.9\columnwidth]{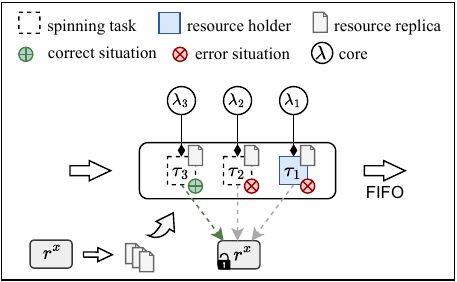}
\caption{Fault tolerance of global critical sections}
\label{fig:check2}
\end{figure}

Additionally, MSRP-FT incorporates a novel fault-tolerance mechanism specifically designed for global resources. As shown in Figure \ref{fig:check2}, when a task requests a global resource, it enters a FIFO queue, where the resource is allocated in FIFO order. The task at the head of the queue (i.e., $\tau_1$) becomes the resource holder, while other tasks remain in a busy-wait state, spinning. The resource holder reads the shared resource and duplicates its execution into replicas, enabling all spinning tasks (i.e., $\tau_2$ and $\tau_3$) to assist in executing its critical sections in parallel. Replicas are executed locally, once a fault is detected, the replica is re-executed on the same core. A resource update is performed as an atomic action and only fault-free executions (i.e., $\tau_3$) are eligible for updating the resource. Once the resource is updated, all replicas are terminated, the head task exits the queue, and the next task in line (i.e., $\tau_2$) repeats the process.

\begin{figure}[t]
    \centering
    \includegraphics[width=0.8\columnwidth]{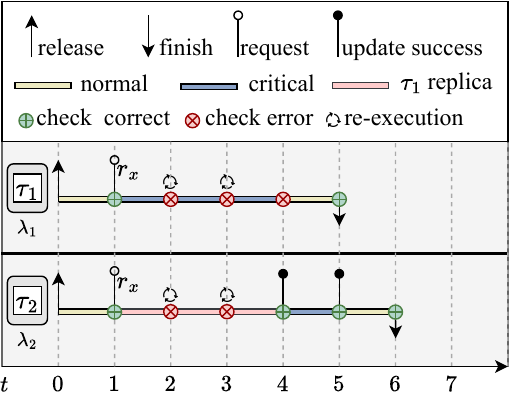}
    \caption{Worst-case resource accessing of $\tau_2$ under  MSPR-FT}
    \label{MSRPFT}
\vspace{-1em}
\end{figure}

The amount of time that a task needs to execute all its  $n^x_j$ executions is calculated as $\lceil n^x_j/k_j \rceil \cdot c^x$ \cite{chen2022msrp}, where $k_j$ represents the number of execution units available (including the task itself and helper tasks positioned behind it in the FIFO queue).
To intuitively illustrate the performance of MSRP-FT, we demonstrate the following example which indicates the worst-case resource accessing time of $\tau_2$.

\begin{example}\label{example:MSRPFT}
As shown in Figure \ref{MSRPFT}, $\tau_1$ with $n^x_1 = 6$ executes on $\lambda_1$, and $\tau_2$ with $n^x_2 = 1$ executes on $\lambda_2$. At time $t = 1$, both tasks simultaneously request access to $r^x$. The critical section length is $c^x = 1$. $\tau_1$ becomes the head of the FIFO queue, with $\tau_2$ following. Both begin executing their critical sections concurrently, with $\tau_2$ assisting by executing a replica of $\tau_1$'s critical section. After repeated failures and retries, $\tau_2$ successfully executes the 6th execution at $t = 4$, updates the resource, and helps $\tau_1$ complete, allowing $\tau_1$ to leave the FIFO queue at $t = 4$ and finally finishes at $t = 5$. At $t = 4$, $\tau_2$ then executes its own critical section. As no other tasks are queued, it executes alone and, with $n^x_2 = 1$, completes successfully and updates at $t = 5$. $\tau_2$ finally finishes at $t = 6$.
\end{example}



Under MSRP-FT, without considering coordination overhead, the worst-case resource-access time for a given task $\tau_i$ (e.g., $\tau_2$ in Example \ref{example:MSRPFT}) is determined by the following two conditions. 
First, $\tau_i$ suffers the maximum possible remote blocking (blocking from remote cores due to global resource contention) by being placed at the end of the FIFO queue, with no external helpers available which means it must execute its $n^x_i$ replicas sequentially on its own core.
Second, the maximum number of remote requests (denoted as $m$) ahead of $\tau_i$ in the FIFO queue are arranged such that tasks with higher execution counts $n^x_j$ appear closer to the end of the queue with fewer helpers. 
Specifically, the first remote request (with the largest $n^x_j$) is assisted by only $k_1 = 2$ cores (including itself and $\tau_i$), the second by $k_2 = 3$, and so on. The worst-case access time of request $\tau_i$ is expressed in Equation~\eqref{eq:OneAccessE} \cite{chen2022msrp}.
All terms are measured in units of $c^x$, the execution length of a single critical section of resource $r^x$.
Specifically, $n^x_i$ is the number of re-executions caused by faults, and $S^x_i$ is the blocking delay due to remote requests.
We further observe that each remote request contributes at least one unit of $c^x$ to the blocking, as $\lceil n^x_j / k_j \rceil \geq 1$. Hence, in the worst case, $S^x_i \geq m$.

\begin{equation}\label{eq:OneAccessE}
E^x_i = (n^x_i+S^x_i ) \cdot c^x,\quad S^x_i \geq m
\end{equation}


Based on the characteristics of MSRP-FT we address the following corollary which will be useful later.

\begin{corollary}\label{coro:MSRP-FT}
Under MSRP-FT, the remote blocking units $S^x_i$ 
suffered by a given task $\tau_i$ accumulates with the increase of the $n^x_j$  of the remote requests. More specifically, if more than one remote request has $\lceil n^x_j/k_j \rceil > 1$, then $S^x_i > m+1$. For example, if a task with $k_j=2$ then incurs 2 faults (i.e. $n^x_j=3$), it satisfies $\lceil 3/2 \rceil = 2 > 1$, thus contributing more than one unit of remote blocking. 
\end{corollary}




\subsection{Limitations of state-of-the-art}\label{sec:limitation}

The MSRP-FT protocol improves the fault tolerance efficiency by leveraging the parallel execution of replicas. While this approach improves fault recovery efficiency, it introduces new challenges as summarized below.



\begin{limitation}\label{limit:lock}
With a high number of faults, managing resource accesses sequentially, even with the helping mechanism, can still result in substantial delays.
\end{limitation}

Although MSRP-FT accelerates fault tolerance by allowing spinning tasks to execute replicas in parallel, each task must still wait for the resource access in order.
According to Corollary~\ref{coro:MSRP-FT}, when a remote task experiences frequent faults (i.e., a high $n^x_j$), it can significantly delay the execution of tasks behind it in the FIFO queue, i.e., increase the remote blocking ($S^x_i\cdot c^x$).  


\begin{limitation}\label{limit:complex}
Utilising spinning tasks to assist requires complicated scheduling and may incur substantial overhead.
\end{limitation}

MSRP-FT introduces unavoidable coordination overhead due to its helping mechanism. When a task is at the head of the FIFO queue, it must construct a shared operation descriptor, we call it \textit{wrap}, which includes a function pointer for the critical section logic, a reference to the shared resource, and any relevant execution context \cite{shi2017implementation,zhao2020complete}. This structure is written to a globally visible memory location such as DRAM or the last level cache, and a readiness flag is set using a memory barrier to ensure cross-core visibility~\cite{losa2017transparent}. This step contributes to the per-head metadata setup cost, denoted $O_{\text{wrap}}$.

Each task behind the head in the queue must monitor the shared readiness flag and, once available, retrieve the wrap, copy the shared resource locally, and execute the head task’s critical section. This cooperative execution introduces the per-replica overhead $O_{\text{replica}}$, which includes coordination, local preparation, and control transfer costs \cite{zhao2020complete,shi2017implementation}. When the task later becomes the head itself, it must construct its own wrap, incurring a one-time setup cost denoted as $O_{\text{self\_wrap}}$. 
Therefore, if a task is preceded by $m$ other tasks in the FIFO queue, the total overhead it incurs for a single resource request is given by:
\begin{equation}
O_{\text{total}}(m) = m \cdot (O_{\text{wrap}} + O_{\text{replica}}) + O_{\text{self\_wrap}}
\label{eq:MSRP-FT_overhead}
\end{equation}

\begin{limitation}\label{limit:faulttype}
Replica-based parallel execution, which relies on identical execution logic and synchronized data at the same time point, may be ineffective against deterministic or common-mode faults.
\end{limitation}

In MSRP-FT, all replicas execute identical operations derived from the same wrap descriptor, leading to highly synchronized control flow and memory access patterns. Consequently, faults triggered by specific execution sequences or timing interference may simultaneously affect all replicas, causing coordinated failure.

\section{LEFT-RS Protocol} 
\label{sec:protocol}

In this section, we propose~\name~ for multicore real-time systems.
The~\name~ introduces novel mechanisms to handle faults in accessing global 
resources more efficiently, aiming to improve the system's schedulability. We first introduce the basic setup which includes the part without global resource sharing. Then, we show how~\name~ manages global resource sharing with faults and its design rationale.

\subsection{Normal and local critical sections}

To avoid transitive errors and prevent unnecessary resource contention caused by re-executing the entire task, we continue to adopt the checkpointing mechanism from MSRP-FT, as illustrated in Section \ref{sec:MSRP-FT}. Checkpoints are placed not only at the beginning and end of each task, but also around critical sections within the task’s execution.
These checkpoints divide execution into normal sections (without shared resources) and critical sections (global and local). 
Each checkpoint, if applicable, detects faults in the preceding segment and saves the state for the next segment.
In cases where a task has few or no critical sections, resulting in large normal sections, checkpoints can be introduced to subdivide them and reduce re-execution range. However, since our focus is on fault tolerance in critical sections, we do not further explore checkpoint placement within normal sections.

If a fault occurs in the normal or local critical section, the system simply applies rollback and re-execution. 
For the management of local resources, we assume a ceiling priority protocol, where a task accessing a resource has its priority raised to the ceiling priority, which is equal to the highest priority among local tasks that access the same resource.







\subsection{Global resource sharing }

In this subsection,
we present how global resource sharing is managed by~\name~ through a series of rules.
The presentation order of these rules does not 
follow the execution order of a task when accessing shared resources, instead, it follows the design rationale. In the end, Example \ref{ex:protocol} will show the overview of the entire execution sequence of tasks following the rules.

The core objective of the protocol is to improve resource access efficiency under fault-tolerant conditions. To achieve this, the protocol aims to allow tasks that complete without faults to update the global resources earlier when others encounter faults. This reduces unnecessary delays and enhances overall system throughput.
Enabling concurrent execution of critical sections is key to achieving this goal. 
This is realized through a lock-free approach, as described in Rule~\ref{rule1}, which is feasible as concurrent reads do not modify the shared state.

\begin{ruledef}\label{rule1}
When requesting a global resource, all tasks are allowed to read the resource simultaneously and proceed to their critical sections locally without acquiring a global lock.
\end{ruledef}

While concurrent execution improves efficiency, it introduces two key challenges that must be addressed: (1) concurrent updates to global resources, which may result in data races, and (2) the use of stale data, where tasks operate on outdated information due to limited update visibility.

The first concurrency issue is mitigated by Rules~\ref{rule2} and~\ref{rule3}. In Rule~\ref{rule2}, the FIFO queue records the order of task requests to regulate the update sequence, rather than enforcing resource accessing order as in MSRP-FT. Rule~\ref{rule3} ensures update priority in FIFO order to provide fairness and restricts resource updates to a single task at a time, thus avoiding race conditions.

\begin{ruledef}\label{rule2}
When a task accesses a global resource, it is added to a FIFO queue.  
\end{ruledef}

\begin{ruledef}\label{rule3}
Upon successful (fault-free) execution, if multiple tasks request an update concurrently, the update is granted in FIFO order, and updates must be performed atomically.
\end{ruledef}

The second concurrency issue is resolved by Rule~\ref{rule4}, which ensures that each update triggers the re-read and re-execution of other tasks in the FIFO queue to avoid outdated execution.

\begin{ruledef}\label{rule4}
When a task successfully updates the global resource, it leaves the FIFO queue. All other tasks operating on outdated copies of the resource are notified to abort their current executions, discard local copies, and restart their critical sections using the updated global resource.
\end{ruledef}

In addition, tasks also re-execute due to transient faults encountered during their critical sections as defined in Rule~\ref{rule5}. This leads to two types of re-execution behaviours of critical sections as shown in Definitions \ref{def:data-induced} and \ref{def:fault-induced}.

\begin{ruledef}\label{rule5}
Upon unsuccessful (faulty) execution, the task re-executes the faulty critical section individually.
\end{ruledef}

\begin{definition}\label{def:data-induced}
[Data-Induced Re-execution]
Re-execution is triggered when another task updates the shared resource, invalidating the current task’s local copy and requiring it to restart with the latest data.
\end{definition}

\begin{definition}\label{def:fault-induced}
[Fault-Induced Re-execution]
Re-execution is triggered when a task detects a transient fault during its critical section, requiring a retry to ensure correctness.
\end{definition}

The main concurrency issues are addressed by regulating concurrent updates and data consistency enforcement.
However, as resource accesses can start and finish at different times due to variations in request arrivals and execution paths, tasks may complete their executions out of order. 
Thus, regulating the timing of updates and re-executions remains critical to prevent unfairness and unpredictable delays. To address this issue, the protocol introduces a comprehensive synchronisation mechanism composed of three tightly integrated rules.
First, Rule~\ref{rule6} ensures that tasks joining the FIFO at different times begin execution simultaneously, aligning their execution windows and reducing initial timing gaps.

\begin{ruledef}\label{rule6}
Upon joining the FIFO queue, if the head task in the queue is in the middle of executing a critical section, subsequent tasks enter a \textit{synchronisation period} and wait for the head task to finish its current execution; otherwise, the synchronisation step is skipped.
\end{ruledef}

However, tasks may still finish their critical sections earlier than expected.
For example, consider a scenario where $\tau_a$, $\tau_b$, and $\tau_c$ are placed in FIFO order from front to end. If $\tau_c$ completes earlier and updates the resource, this may trigger data-induced re-executions of $\tau_a$ and $\tau_b$. Furthermore, later-joining tasks may starve $\tau_a$ and $\tau_b$ for the same reason.
To address this issue, Rule~\ref{rule7} governs the update of a shared resource by tasks upon successful execution. Even if a task finishes earlier, it cannot update the global resource until all preceding tasks have either been completed or encountered faults. This prevents early completion from triggering unnecessary data-induced re-execution of preceding tasks and ensures update fairness, which completes Rule \ref{rule3}.

\begin{ruledef}\label{rule7}
Upon successful (fault-free) execution, the head task in the FIFO queue can directly attempt to update the shared resource. Other tasks (apart from the head) can attempt to update only after confirming that all preceding tasks have encountered faults during their current execution.
\end{ruledef}

In this scenario, suppose $\tau_c$ completes successfully. According to Rule~\ref{rule7}, it must wait for $\tau_a$ and $\tau_b$ to finish before updating. If $\tau_a$ completes first but encounters faults and immediately re-executes, then when $\tau_b$ eventually finishes, $\tau_c$ must still wait for $\tau_a$'s new execution, which introduces additional delays. Therefore, Rule~\ref{rule8} complements Rule \ref{rule5} by controlling fault-induced re-executions timing.  This prevents re-executing tasks from interfering with later tasks' updates, avoiding serial delays from unsynchronised retries.

\begin{ruledef}\label{rule8}
Upon unsuccessful (faulty) execution, the task waits for all tasks in the FIFO queue to finish their current executions and then re-executes the faulty critical section by itself.
\end{ruledef}

These design elements work collectively to realize two core rationales that govern the overall protocol behaviour.

\begin{itemize}
\item \textit{Rationale 1}: A task is re-executed due to the update of the global resource by preceding tasks in the FIFO queue. 
\item \textit{Rationale 2}: A task is re-executed due to an update made by later-arriving tasks. However, such re-executions overlap with fault-induced re-executions.
\end{itemize}

Rationale 1 ensures fairness and bounded re-execution: tasks re-execute due to their relative positions in the FIFO queue. Under Rule~\ref{rule4}, each update to the global resource triggers the departure of a task from the FIFO queue. Therefore, the number of re-executions is predictable. 
Rationale 2 is the most critical for achieving efficient resource access under fault tolerance. If a preceding task encounters faults, it will re-execute regardless. If this re-execution overlaps with the resource update made by a later-arriving task, it benefits the later-arriving task by reducing delay without adversely affecting the preceding task. Since the forced waiting time occurs only when tasks complete earlier than expected, the total of execution and waiting time still falls within the duration of a single critical section, which will be further proved in Lemma~\ref{lemma:onesync}.

Finally, Rules \ref{rule9} and \ref{rule10} complete the protocol.
Rule \ref{rule9}  can be implemented by raising its active priority to the system’s maximum level, ensuring it is not preempted during global resource access.
This behaviour is inherited from MSRP~\cite{gai2001minimizing}, chosen for its simplicity. While ceiling-based approaches, such as MrsP~\cite{burns2013schedulability}, may further reduce arrival blocking, they are beyond the scope of this study.

\begin{ruledef}\label{rule9}
When a task requests a global resource, it becomes non-preemptive with respect to local tasks and remains so until it completes the resource access.
\end{ruledef}

Rule \ref{rule10} completes Rule \ref{rule6} by eliminating unnecessary delays. 
We propose Lemma \ref{lemma:synfree} to support Rule \ref{rule10}.

\begin{ruledef}\label{rule10}
The initial synchronisation period may be skipped if a task joins the FIFO queue with the assurance that all preceding tasks will not incur faults (i.e., $n^x_j = 1$, or they have already encountered the maximum number of allowable faults), or if the system's fault-tolerance mode is disabled.
\end{ruledef}

\begin{lemma}\label{lemma:synfree}
    For a given task  $\tau_i$, if its preceding tasks in the FIFO queue will not incur faults, then the request of  $\tau_i$ does not need to enter the synchronisation period.
\end{lemma}

\begin{proof}

As the preceding tasks must start no later than  $\tau_i$, a resource update must be made within an execution period $c^x$
from preceding tasks since they will not encounter faults. Even if $\tau_i$ finishes earlier, 
Rule~\ref{rule7} is sufficient to guarantee that the task can only update the resource when the preceding tasks finish, which aligns with Rationale 1 and 2.
\end{proof}

\begin{figure}[t]
\centering
\includegraphics[width=\columnwidth]{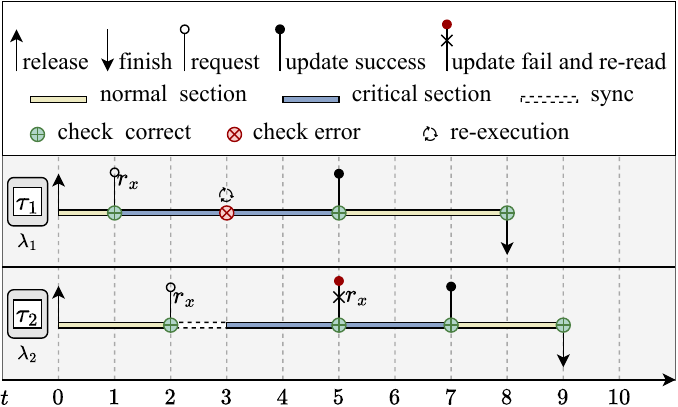}
\caption{Example of global resource management of LEFT-RS} 
\label{fig:management}
\vspace{-1em}
\end{figure}

These rules are visualized by Example \ref{ex:protocol} which illustrates the complete execution sequence of a task request for a global resource under LEFT-RS.

\begin{example}\label{ex:protocol}
As shown in Figure \ref{fig:management}, $\tau_1$ and $\tau_2$ are executing on cores $\lambda_1$ and $\lambda_2$, respectively. At $t=1$, $\tau_1$ requests access to $r^x$ with $c^x = 2$. It becomes locally non-preemptive and starts executing its critical section immediately (Rules \ref{rule1} and \ref{rule9}) without a synchronisation period (Rule \ref{rule6}). 
At $t=2$, $\tau_2$ requests access to $r^x$ and the current FIFO queue order becomes $\tau_1$ and $\tau_2$, from front to end. According to Rule \ref{rule6}, since $\tau_1$ is in the middle of its critical section execution at $t=2$, $\tau_2$  enters a synchronisation period. 
At $t=3$, as $\tau_1$ fails its execution and all executions are synchronised, $\tau_1$ re-executes (Rules \ref{rule5} and \ref{rule8}) and $\tau_2$ reads the shared resource and begins executing its critical section in parallel with $\tau_1$ (Rule \ref{rule1}). At $t=5$, the executions of $\tau_1$ and $\tau_2$ are correct and synchronised, $\tau_1$ successfully updates $r^x$, in accordance with the update procedure (Rules \ref{rule3} and \ref{rule7}). At $t=5$, after the resource update, $\tau_1$ leaves the FIFO queue (Rule \ref{rule4}) and completes its execution at $t=8$. Also, at $t=5$, $\tau_2$ is signalled to re-read and re-execute the shared resource (Rule \ref{rule4}). At $t=7$, $\tau_2$ updates $r^x$ successfully (Rules \ref{rule3} and \ref{rule7}), leaves the FIFO queue (Rule \ref{rule4}), and finishes at $t=9$.
\end{example}

Overall, LEFT-RS allows tasks to start resource accessing concurrently without being served in FIFO order. Each task can exit the FIFO queue earlier upon successful execution if the preceding tasks incur faults. Such a design can effectively address Limitation~\ref{limit:lock}. Moreover, ~\name~ avoids the coordination overhead described in Equation~\eqref{eq:MSRP-FT_overhead} by eliminating cooperative replica execution. Tasks do not construct or publish shared wrap descriptors, nor do they assist in executing other tasks' critical sections. Instead, each task executes its own critical section independently after retrieving the shared resource state which addresses Limitation \ref{limit:complex}.
Other overheads, such as local fault detection and thread communications, are omitted from the discussion since they are either common to both protocols or have negligible performance impact. Since tasks under ~\name~ operate on their own logic and execution paths, faults affecting one task are less likely to propagate to others. This design inherently reduces the risk of synchronised failure and improves fault isolation across concurrent executions which addresses Limitation \ref{limit:faulttype}.
The effectiveness of the proposed rules and rationales will be further demonstrated in the worst-case analysis of resource access in Section~\ref{sec:performance}.

\subsection{Implementation discussion}

The lock-free approach is well-established and has a long history \cite{anderson1997real}. LEFT-RS is applicable to shared resources where tasks can read a consistent snapshot and perform local computation before a single atomic update. Typical examples include status flags, control registers, configuration parameters, and shared counters. These structures are common in embedded systems and align with LEFT-RS’s design, which favours safe concurrent reads and avoids complex coordination. Furthermore, both LEFT-RS and MSRP-FT \cite{chen2022msrp} require atomic updates and execute critical sections concurrently and locally, meaning their theoretical applicability spectra are the same.

LEFT-RS requires each task to copy the shared resource into its private memory (e.g., L1/L2 cache) and execute independently, which is supported by most commercially off-the-shelf (COTS) architectures. From a software perspective, Rule \ref{rule4} involves signalling other tasks, while Rules \ref{rule6}, \ref{rule7}, and \ref{rule8} involve checking the execution status of other tasks. Such coordination can be achieved through lightweight communication mechanisms, such as shared status flags or conditional signalling.

\section{Theoretical Comparison with MSRP-FT} 
\label{sec:performance}

The key distinction between the proposed approach and MSRP-FT lies in the management of global resource access.
Before carrying out the full analysis,
in this section, we analyse the advantages of the proposed approach by examining the worst-case resource accessing time for one single request of a given task $\tau_i$ to a global resource $r^x$.


According to Rule \ref{rule5}, 
each task under the proposed approach 
re-executes independently without assistance when it incurs faults. 
Therefore, in the worst-case scenario without accounting for the delay caused by remote tasks, a resource request from $\tau_i$ will execute $n^x_i$ times its critical sections, consisting of $f^x_i$ failed executions and one successful submission.

According to Rules \ref{rule4} and \ref{rule7}, each task will perform data-induced re-execution caused by 
the resource updates from preceding tasks in the FIFO queue. Each resource update will lead to the leave of a task from the FIFO queue. Therefore, the amount of data-induced re-execution can be bounded by the following lemma.


\begin{lemma}\label{lemma:outdated}
For a given request from a task $\tau_i$ to $r^x$, 
if there are at most $m$ ($\leq |\Lambda|$ ) remote requests that can request $r^x$ at the same time, then the request of $\tau_i$ can incur at most $m$ instances of data-induced re-execution.
\end{lemma}

\begin{proof}
Data-induced re-execution occurs due to a successful resource update in the system. In the meantime, each update removes a task from the FIFO queue (Rule \ref{rule4}). There are at most $m$ instances of preceding tasks in the queue, therefore, the request of  $\tau_i$ can incur at most $m$ data-induced re-executions from preceding tasks. According to Rule \ref{rule7}, later-arriving tasks can only cause data-induced re-execution to $\tau_i$ when it incurs faults. This is already accounted for in the $f^x_i$ failed executions due to faults.
\end{proof}

According to Rule \ref{rule6}, each task may enter the synchronization period when requesting a resource. We propose the following lemma to bound the synchronization period.

\begin{lemma}\label{lemma:overhead}
The maximum synchronization overhead for a given resource request is one unit of $c^x$.
\end{lemma}

\begin{proof}
A resource request incurs synchronization overhead when it joins the FIFO queue and waits for the head task to complete its current execution. Since the maximum execution length is $c^x$ and synchronization occurs at most once, the maximum synchronization overhead is $c^x$.
\end{proof}

\begin{lemma}\label{lemma:onesync}
The synchronisation overhead is the only additional waiting time.
\end{lemma}

\begin{proof}
According to Rule \ref{rule7}, a task may need to wait for the preceding tasks to finish when executed correctly. However, Rule \ref{rule6} forces all tasks to start simultaneously; therefore, the waiting period only occurs when the task finishes earlier, and the sum of its execution and waiting time will not exceed $c^x$. The same reason also applies to Rule \ref{rule8}, when a task incurs faults, it waits for all tasks to finish before carrying out re-execution.
\end{proof}


According to Lemmas \ref{lemma:outdated}, \ref{lemma:overhead} and \ref{lemma:onesync}, we can have the following corollary. In Corollary~\ref{cor:proposed}, $n^x_i$ is the fundamental worst-case execution of $\tau_i's$ request without accounting for the delay from remote cores. The $m$ accounts for data-induced re-executions as demonstrated in Lemma  \ref{lemma:outdated}. The 1 unit of $c^x$ is the synchronization overhead as illustrated in Lemma \ref{lemma:overhead}.

\begin{corollary}\label{cor:proposed} 
For a given resource request to $r^x$ from $\tau_i$ under LEFT-RS, its worst-case resource accessing time can be upper bounded as $E^x_i=(n^x_i+m+1) \cdot c^x$.
\end{corollary}







According to Lemma~\ref{lemma:synfree}, the synchronization overhead may not happen when preceding tasks of $\tau_i$ in the FIFO queue do not encounter faults, and we can have the following corollary.

\begin{corollary}\label{cor:MSPP} For a given resource request of $ \tau_i$ to $r^x$,
if preceding tasks in the FIFO queue incur no faults, the worst-case resource accessing time under the proposed approach will be $E^x_i =(n^x_i+m) \cdot c^x$.
\end{corollary}


As shown in Equation~\eqref{eq:OneAccessE}, the worst-case resource accessing time for a given request from $\tau_i$ under MSRP-FT is 
$E^x_i \geq ( n^x_i+m) \cdot c^x$,
which indicates that it can be better (i.e., smaller) than the proposed bound 
$
E^x_i = ( n^x_i+m + 1) \cdot c^x
$
in limited scenarios.
More specifically, according to Corollary \ref{coro:MSRP-FT} and \ref{cor:proposed}, we can have the following corollary.

\begin{corollary}\label{cor:better_performance} 
For a request from $\tau_i$ to $r^x$, 
LEFT-RS outperforms MSRP-FT as the number of faults increases. More specifically, this holds when more than one remote request satisfies $\lceil n^x_j/k_j \rceil > 1$. Unlike MSRP-FT, LEFT-RS maintains a fixed number of remote delay units at $m$, ensuring that its remote blocking does not grow with the fault count.
\end{corollary}

To visualize the effectiveness of the proposed approach, we illustrate the following example with the same setup of Example \ref{example:MSRPFT} which illustrates the worst-case resource accessing time of $\tau_2$ as well.

\begin{example}\label{example:LEFT-RS}
As shown in Figure \ref{fig:LEFTRSCOMPARE}, $\tau_1$ with $n^x_1 = 6$ executes on $\lambda_1$, and $\tau_2$ with $n^x_2 = 1$ executes on $\lambda_2$. At time $t = 1$, both tasks simultaneously request access to $r^x$. 
Under LEFT-RS, both tasks begin executing their own critical sections concurrently (Rule \ref{rule1}). They both complete successfully at $t = 2$, with $\tau_1$ updates $r^x$, while $\tau_2$ is blocked (Rule \ref{rule3}). This represents the worst case for $\tau_2$, as it would have updated  $r^x$ if $\tau_1$ had failed. At $t = 2$, $\tau_2$ re-reads $r^x$ and re-executes (Rule \ref{rule4}). It completes accessing $r^x$ at $t = 3$ and finishes at $t = 4$, which is earlier than the $t = 6$ completion in Example \ref{example:MSRPFT}.
\end{example}

\begin{figure}[t]
\centering
\includegraphics[width=0.8\columnwidth]{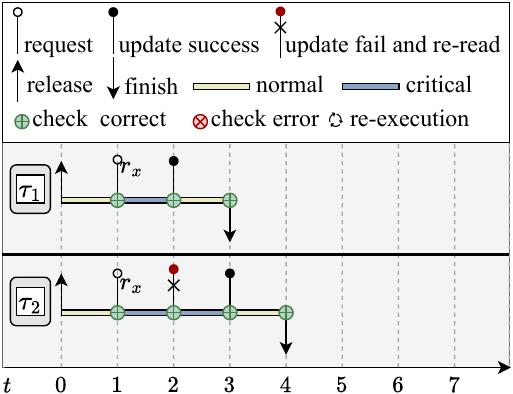}
\caption{Worst-case resource accessing of $\tau_2$ with LEFT-RS}
\label{fig:LEFTRSCOMPARE}
\vspace{-1em}
\end{figure}

In addition, the  Corollary \ref{cor:better_performance} and Example \ref{example:LEFT-RS} are based on the theoretical worst-case response time analysis of MSRP-FT as introduced in \cite{chen2022msrp}. The unavoidable overhead associated with MSRP-FT, as described in Equation~\eqref{eq:MSRP-FT_overhead}, has not yet been taken into account.
The overall performance comparisons will be demonstrated in Section~\ref{sec:evaluation}.

\section{Schedulability Analysis} 
\label{sec:analysis}

In this section, we derive the worst-case response time analysis of LEFT-RS. The notations specific to this analysis are listed in Table \ref{tab:schedulability_notation}, while general notations reused from earlier sections can be found in Table \ref{tab:notation}.

As shown in Equation \eqref{eq:Respnse}, $R_i$ represents the worst-case response time of a given task $\tau_i$. $C_i$ represents the pure WCET of $\tau_i$ without accessing any shared resource. $E_i$ is the total resource-accessing time of $\tau_i$ and local tasks with higher priority than $\tau_i$ ($lhp(i)$). $B_i$ is the arrival blocking incurred by $\tau_i$ upon arrival.
$F_i$ denotes the amount of time $\tau_i$ spends tolerating faults.
The term $\left \lceil R_i/T_h \right \rceil$ denotes the maximum number of times a local high-priority task $\tau_h\in lhp(i)$ can preempt $\tau_i$ during $R_i$. $C_h$ and $F_{h}$ are the pure WCET and fault-tolerance time of $\tau_h$ per release instance, respectively. 
\begin{equation} \label{eq:Respnse}
\begin{split}
R_i =  
C_{i} + E_{i} + B_{i}+F_i+\sum_{\tau_h \in {lhp}(i)} \left \lceil \frac{R_i}{T_h}\right \rceil \cdot (C_{h}+ F_{h})
\end{split}
\end{equation}

We first introduce how to bound $E_i$
which includes the resource-accessing time of  $\tau_i$ and tasks in $lhp(i)$, as in the worst-case tasks in $lhp(i)$
preempt $\tau_i$ and finish their resource accessing during preemption which transitively delays $\tau_i$ \cite{zhao2017new}. As shown in Equation \eqref{eq:Ei},  $E_i$ iterates over all the shared resources in the system ($\mathcal{R}$). For a given resource $r^x$,
$ N^x_{i,local}$ represents the total number of requests from $\tau_i$ and tasks in $lhp(i)$ to $r^x$ and they are called local requests. 
We made an assumption here that every local request has an execution $n^x_i=1$ because the worst-case fault-tolerance time is accounted in $F_i$ and $F_h$.
$ N^x_{i,local}$ is further expanded in Equation \eqref{eq:Nlocal}.
$|\mathcal{E}^x_i|$ gives the total number of remote blocking requests incurred by local requests which is illustrated through steps in Equations \eqref{eq:S}, \eqref{eq:xii} and \eqref{eq:Nmax}. $ Syn^x_i$ is the synchronisation overhead experienced by local requests as explained in Equation \eqref{eq:synOverhead}. 

\begin{equation}\label{eq:Ei}
E_i = \sum_{r^x \in \mathcal{R}}( N^x_{i,local}+ |\mathcal{E}^x_i|+ Syn^x_i)\cdot c^x
\end{equation}

As shown in Equation \eqref{eq:Nlocal},
$N^x_{i,local}$ includes the total number of requests from $\tau_i$ to $r^x$ ($N^x_i$). During the release time of $\tau_i$ (denoted by $R_i$), $\tau_i$ may be preempted by tasks in $lhp(i)$, it additionally accounts for the requests from all $\tau_h \in lhp(i)$ by including the corresponding $N^x_h$ of each $\tau_h$.


\begin{equation}\label{eq:Nlocal}
N^x_{i,local} = N^x_{i} + \sum_{\tau_h \in {lhp}(i)} \left \lceil \frac{R_i}{T_h}\right \rceil \cdot N^x_{h}
\normalsize
\end{equation}

\begin{table}[t]
    \centering
    \renewcommand{\arraystretch}{1.1} 
    \setlength{\tabcolsep}{8pt} 
    \vspace{1em}
    \caption{Table of Notation for Schedulability Analysis}
    \begin{tabular}{l  p{5.8cm}}
        \hline
        \textbf{Notations} & \textbf{Descriptions} \\
        \hline
        $R_i$ & Worst-case response time of $\tau_i$ \\
        $E_i$ & Total resource-accessing time of $\tau_i$ and $lhp(i)$ \\
        $B_i$ & Arrival blocking incurred by $\tau_i$ \\
        $F_i$, $F_h$ & Fault-tolerance time for $\tau_i$ and $\tau_h \in lhp(i)$ \\
        \hline
        $lhp(i)$, $llp(i)$ & Set of higher/lower-priority tasks relative to $\tau_i$ on the same core\\
        $\lambda(\tau_i)$ & Core where $\tau_i$ is allocated \\
        $\Gamma(\lambda_k)$ & Tasks allocated to core $\lambda_k$ \\
        \hline
        $N^x_{i,local}$ & Number of requests issued by $\tau_i$ and $lhp(i)$ to $r^x$ \\
        $\eta^x_{i}(R_i, R_j)$ & List of $\tau_j$'s requests to $r^x$ over time periods $R_i$ and $R_j$ \\
        $\xi^x_{i,\lambda_k}$ & Remote requests ($n^x_j$) that can block $\tau_i$ and $lhp(i)$ from $\lambda_k$ \\
        $\mathcal{N}^x_{i,\lambda_k}$ & Maximum number of remote requests blocking $\tau_i$ and $lhp(i)$ from $\lambda_k$ \\
        $\mathcal{E}^x_i$ & Total list of $n^x_j$ of remote requests that can block $\tau_i$ and $lhp(i)$ \\
        \hline
        $Syn^x_i$ & Synchronisation overhead incurred by local requests to $r^x$\\
        $\alpha^x_{i}$ & Largest $n^x_j$ from $llp(i)$ for $r^x$ \\
        $\beta^x_{i}$ & The list of remote $n^x_j$ that can block $\alpha^x_{i}$ \\
         \( \mathbb{I}(\cdot) \)  &  An indicator function that returns 1 if the condition inside holds true and 0 otherwise\\
        $F^A(i)$ & Resources imposing arrival blocking on $\tau_i$ \\
        $F(\tau_i)$  & Resources accessed by $\tau_i$\\
        \hline
    \end{tabular}
    \vspace{-4pt}
    \label{tab:schedulability_notation}
\end{table}

$|\mathcal{E}^x_i|$ denotes the size of $\mathcal{E}^x_i$ which is worst-case list of remote requests in terms of $n^x_j$ that can block local requests.   In Equation \eqref{eq:S}, $\lambda(\tau_i)$ denotes the core where $\tau_i$ is allocated.
$\mathcal{E}^x_i$ iterates over all the remote cores of $\tau_i$ (i.e., $\lambda_k \in \Lambda,\lambda_k \neq \lambda(\tau_i)$ )  and takes the first $\mathcal{N}^x_{i,\lambda_k}$   requests from $\xi^x_{i, \lambda_k}$.

\begin{equation}\label{eq:S}
\mathcal{E}^x_i = \bigcup_{\lambda_k \in \Lambda,\lambda_k \neq \lambda(\tau_i)} \bigcup_{p=1}^{\mathcal{N}^x_{i,\lambda_k}} \xi^x_{i, \lambda_k}(p) 
\end{equation}

$\xi^x_{i, \lambda_k}$ as shown in Equation~\eqref{eq:xii} represents a non-increasing list of total requests in terms of $n^x_j$  from
a given remote core $\lambda_k$.
It iterates over tasks on $\lambda_k$ (denoted as $\Gamma(\lambda_k)$) and takes $n^x_j$ of each request to $r^x$. The
notation  $\eta^x_{j} (R_i, R_j)$ denotes the list of $n^x_{j}$ of requests from a task $\tau_j$ executing on a remote core $\lambda_k$ during the period $R_i$ and $ R_j$, which accounts for the back-to-back hit~\cite{zhao2017new}.
It follows that  
$|\eta^x_{j}(R_i,R_j)| = \left \lceil  (R_i+R_j)/T_{j} \right \rceil \cdot N^x_j$
where $|~\cdot~|$ denotes the size of the list (i.e., the number of requests), $\left \lceil  (R_i+R_j)/T_{j} \right \rceil$ represents the number of times $\tau_j$ is released, and $N^x_j$ is the number of times $\tau_j$ accesses $r^x$ in a single release.   For instance,  
$\eta^x_{j}(R_i,R_j) = \{2,2,2\}$ indicates that there are three requests  with $N^x_j=3$ and  all requests have execution numbers $n^x_{j}=2$.

\begin{equation}\label{eq:xii}
\xi^x_{i,\lambda_k} =  \bigcup_{\tau_j \in \Gamma(\lambda_k),\lambda_k \neq \lambda(\tau_i)}  \eta^x_{j} (R_i, R_j)    
\end{equation}

As each core can manage a request at a time, each local request can have at most one request from a remote core to place ahead of it in the FIFO queue \cite{zhao2017new}. Therefore the maximum of requests from $\lambda_k$ that can block local requests is denoted as $\mathcal{N}^x_{i,\lambda_k} $ as shown in Equation \eqref{eq:Nmax} which equals to the minimum number between local requests and the number of remote requests on a given remote core $\lambda_k$. As $\xi^x_{i, \lambda_k}$ is a non-increasing list, the first  $\mathcal{N}^x_{i,\lambda_k}$ items are the biggest $n^x_j$, which can result in the worst-case synchronisation overhead according to Lemma \ref{lemma:synfree}.

\begin{equation}\label{eq:Nmax}
\mathcal{N}^x_{i,\lambda_k} =Min\{N^x_{i,local} ,|\xi^x_{i, \lambda_k}|\}
\end{equation}



According to Lemma \ref{lemma:synfree}, a request does not incur synchronisation overhead if all preceding requests have \( n^x_j = 1 \).  
The synchronisation overhead of local requests that must be accounted for is summarized in Equation \eqref{eq:synOverhead}, where \( \mathbb{I}(\cdot) \) is an indicator function that returns 1 if the condition inside holds true and 0 otherwise. The first term in the equation represents the total number of preceding requests with \( \mathcal{E}^x_i(p) > 1 \). The second term, \( N^x_{i,local} \), denotes the maximum possible synchronisation overhead as each local request can suffer at most one synchronisation overhead according to Lemma \ref{lemma:overhead}. The overall synchronisation overhead is then determined as the minimum of these two quantities.  

\begin{equation}\label{eq:synOverhead}
Syn^x_i = Min \left\{ \sum_{p=1}^{\left| \mathcal{E}^x_i\right|} \mathbb{I} \left( \mathcal{E}^x_i(p) > 1 \right), N^x_{i,local} \right\}
\end{equation}




The arrival blocking is denoted as $B_i$ as shown in Equation~\eqref{eq:blocking}. As the arrival blocking can only occur once upon the arrival of $\tau_i$ due to a local low-priority task $\tau_l\in llp(i)$ accessing a shared resource with a ceiling priority higher than the priority of $\tau_i$ or it is a global resource (non-preemptive).
The set of resources that can impose arrival blocking to  $\tau_i$ is denoted as $F^A(i)$.  Then $B_i$ iterates over all resources in $F^A(i)$ and finds out the maximum arrival blocking.

\begin{equation}\label{eq:blocking}
    B_i = \max_{r^x \in F^A(i)}
    \left\{
        (\alpha^x_{i} + |\beta^x_{i}| + \mathbb{I}(\exists  n^x_{j} \in \beta^x_{i},\ n^x_j > 1) )\cdot c^x
    \right\}
\end{equation}

The first term $\alpha^x_{i}$ in Equation \eqref{eq:blocking}  is expanded in Equation~\eqref{eq:alpha}, which denotes the largest $n^x_l$ of a request among requests of local low-priority tasks ($ \tau_l \in llp(i)$) to a given resource $r^x\in F^A(i)$ (expressed as $N^x_l > 0$).

\begin{equation}\label{eq:alpha}
    \alpha^x_{i} = \max \{ n^x_{l} | \tau_l \in llp(i) \wedge N^x_l >0 \}
\end{equation}

 The second term indicates the size of $\beta^x_{i}$, which denotes the list of remote blocking requests incurred by $\alpha^x_{i}$  shown in Equation \ref{eq:Bi}. $\beta^x_{i}$ iterates over all the remote cores ( $\lambda_k \in \Lambda, \lambda_k \neq \lambda(\tau_i)$) and takes the $\mathcal{N}^x_{i,\lambda_k} + 1$ request from $ \xi^x_{i, \lambda_k}$  if the length is satisfied ($\mathcal{N}^x_{i,\lambda_k}  + 1 \leq |\xi^x_{i, \lambda_k}|$).  This is because the first $\mathcal{N}^x_{i,\lambda_k}$ is already taken by $E_i$ through Equation \eqref{eq:S}. 

\begin{equation}\label{eq:Bi}
    \beta^x_{i} = \bigcup_{\lambda_k \in \Lambda, \lambda_k \neq \lambda(\tau_i)} 
    \left\{ \xi^x_{i, \lambda_k}(\mathcal{N}^x_{i,\lambda_k} + 1) \mid \mathcal{N}^x_{i,\lambda_k}  + 1 \leq |\xi^x_{i, \lambda_k}| \right\}
\end{equation}

Finally, $\mathbb{I}(\exists  n^x_{j} \in \beta^x_{i},\ n^x_j > 1)$ checks whether any request in $\beta^x_{i}$ has $n^x_j > 1$. If so, the function returns 1, indicating that an additional synchronisation overhead must be added as previously described; otherwise, it returns 0.


As the checkpoints divide tasks into segments and each task is assumed to incur at most $f_i$ faults during the execution of the single release, the worst-case fault-tolerance scenario therefore is when all $f_i$ faults happen on the longest segment of a task repeatedly  \cite{punnekkat2001analysis}.
Therefore, the worst-case fault-tolerance time for $\tau_i$ and $lhp(i)$ is denoted in Equation~\eqref{eq:F1}. It identifies the longest segment by taking the maximum value between the total execution time of normal sections $C_i$ and $c^x$, which is the maximum critical section among the resources accessed by $\tau_i$ denoted as $F(\tau_i)$.
The worst-case fault-tolerance time is calculated by  $f_i$  times the largest segment as each task re-executes by itself when encounters faults under LEFT-RS. 

\begin{equation} \label{eq:F1} 
\small F_i = f_{i}\cdot\max\{C_{i},\max\{c^x|r^x\in F(\tau_i)\}\}\ \normalsize
\end{equation}

\section{Evaluation} 
\label{sec:evaluation}

\begin{figure*}[t]

    \begin{subfigure}[b]{0.33\textwidth}
        \centering
        \includegraphics[width=\linewidth]{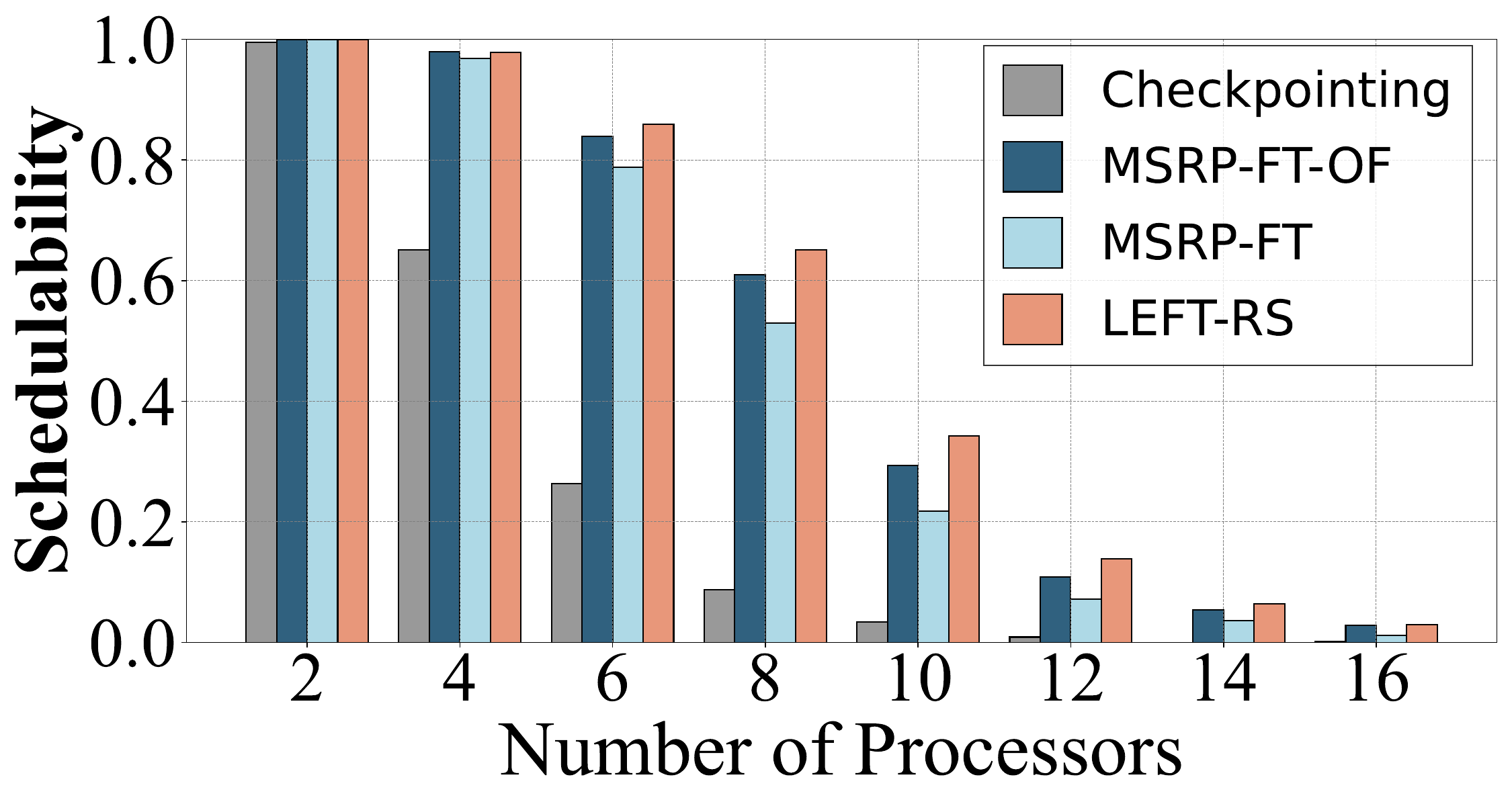}
        \caption{ Schedulability with varied $M$}
        \label{a3}
    \end{subfigure}
    \begin{subfigure}[b]{0.33\textwidth}
        \centering
        \includegraphics[width=\linewidth]{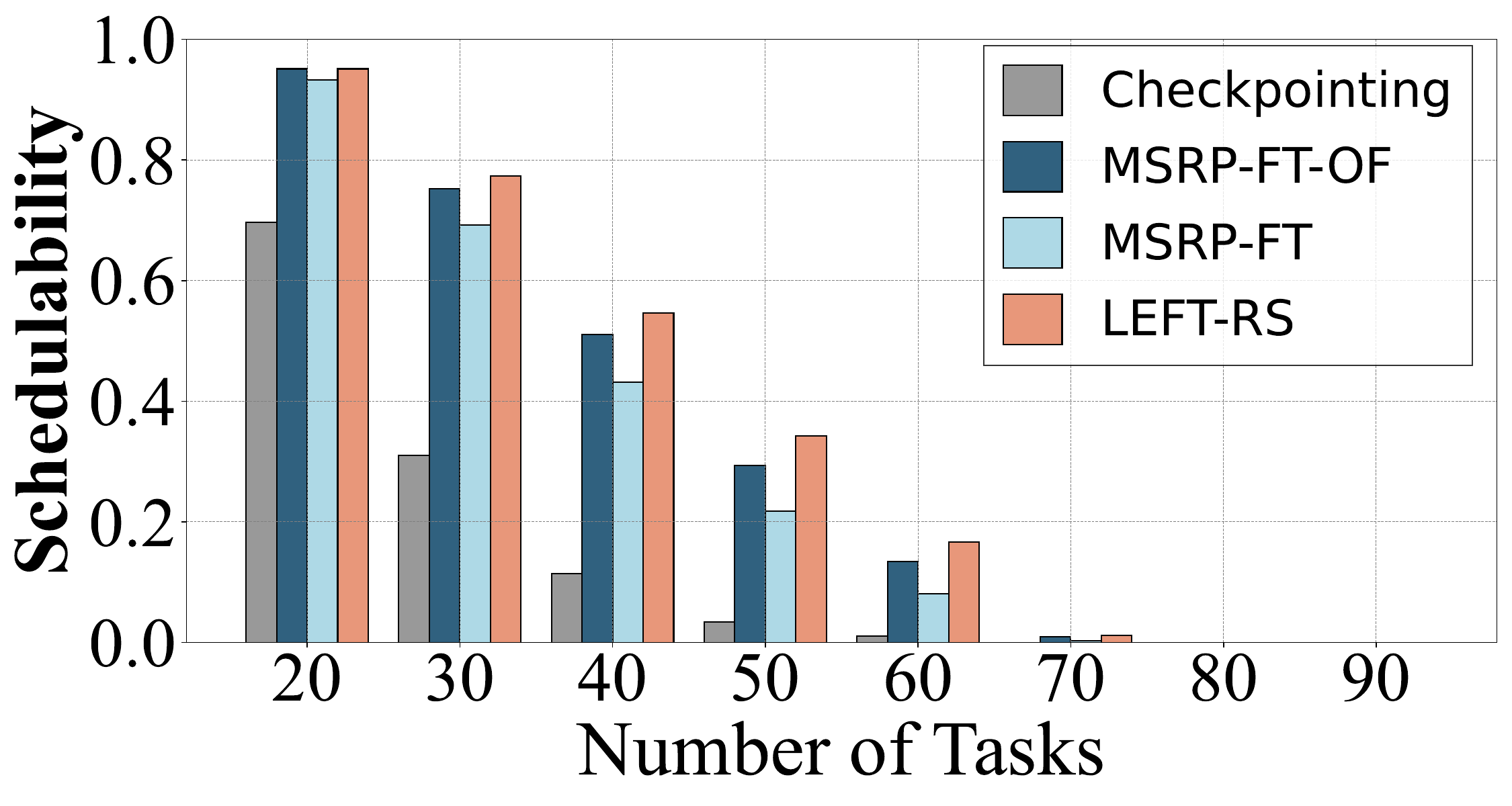}
        \caption{  Schedulability with varied $N\cdot M$}
        \label{b3}
    \end{subfigure}
    \begin{subfigure}[b]{0.33\textwidth}
        \centering
        \includegraphics[width=\linewidth]{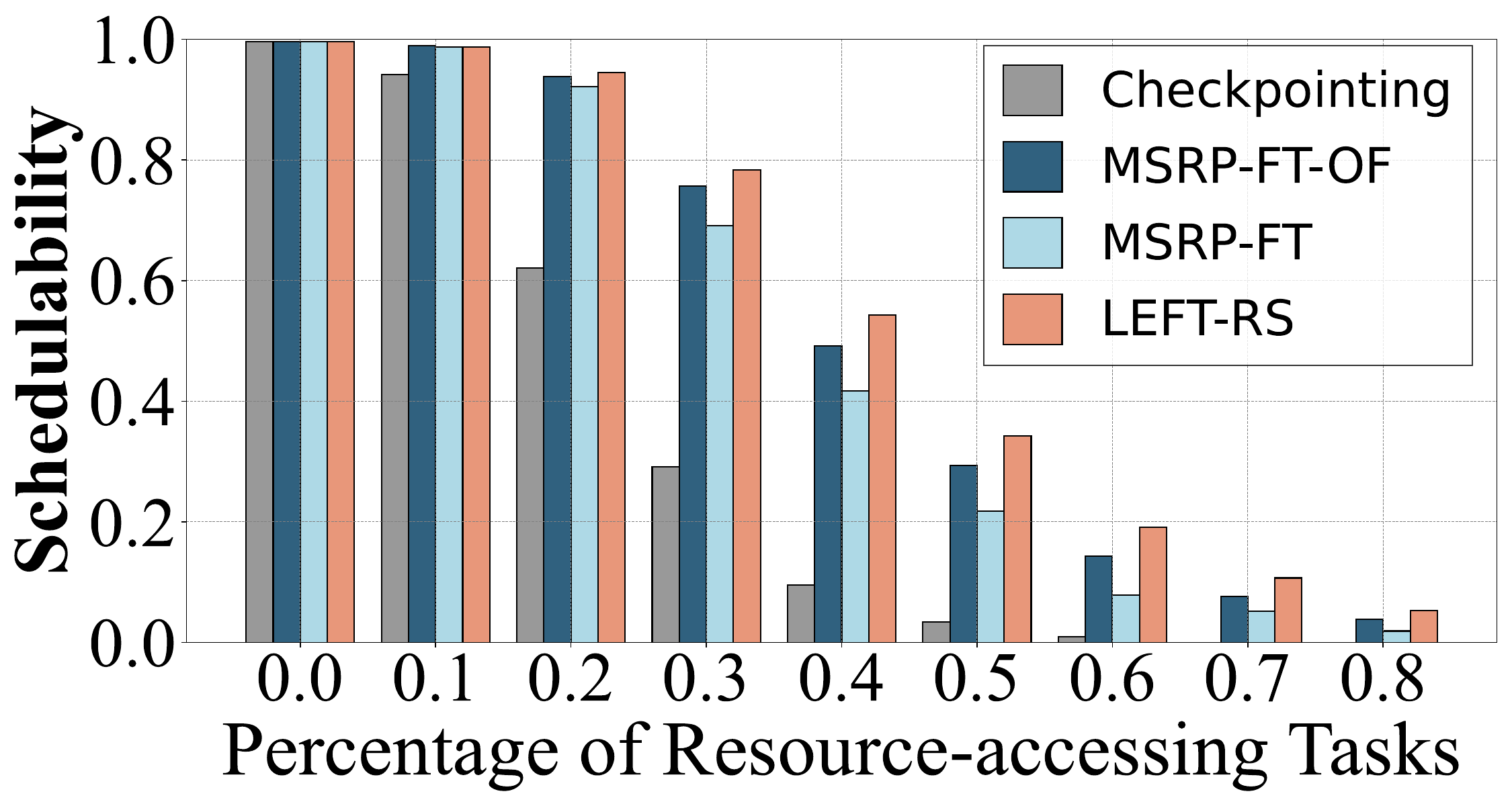}
        \caption{ Schedulability with varied $rsf$}
        \label{c3}
    \end{subfigure}

    \vskip\baselineskip  

    \begin{subfigure}[b]{0.33\textwidth}
        \centering
        \includegraphics[width=\linewidth]{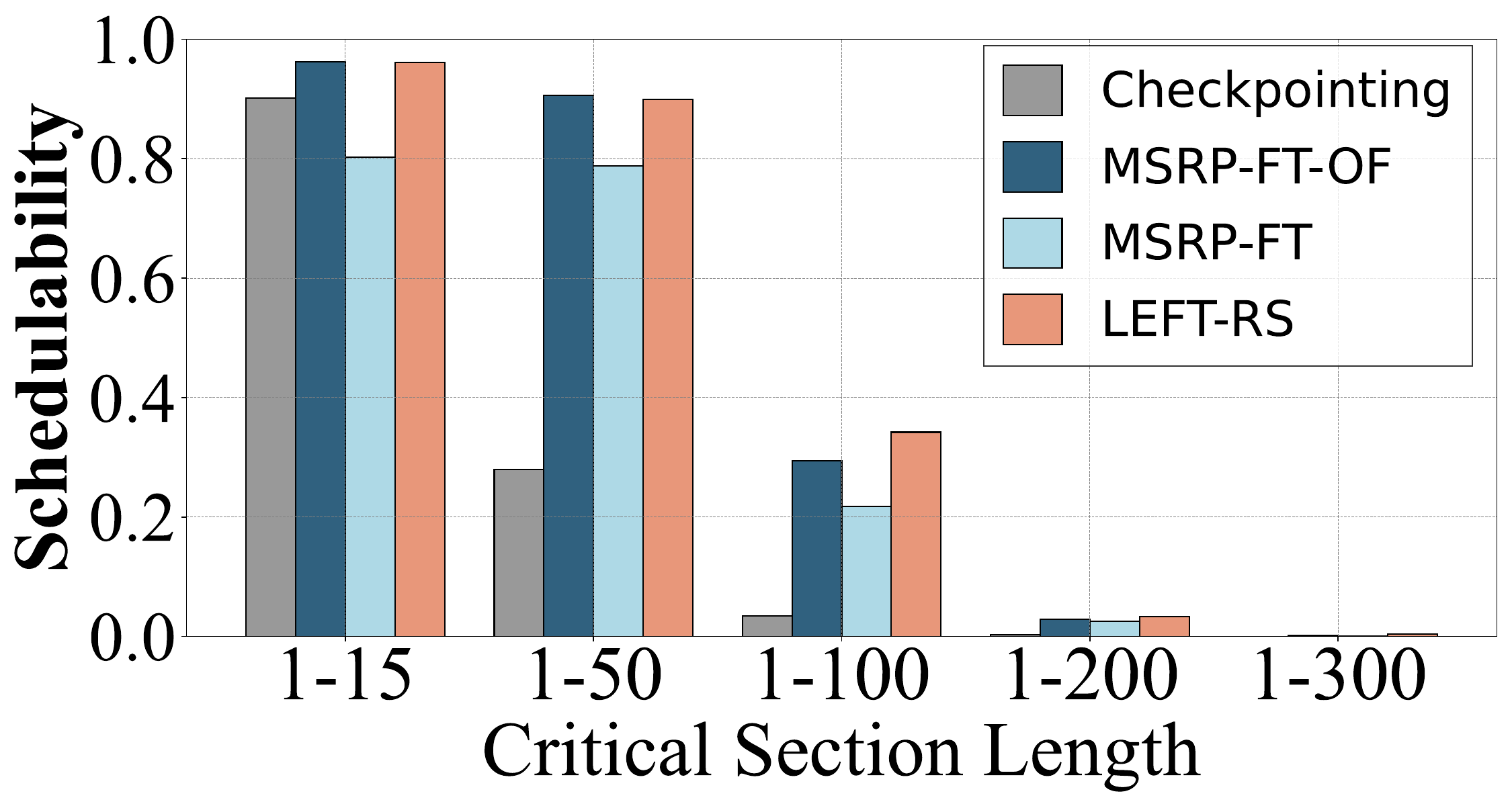}
        \caption{ Schedulability with varied $L$}
        \label{d3}
    \end{subfigure}
    \begin{subfigure}[b]{0.33\textwidth}
        \centering
        \includegraphics[width=\linewidth]{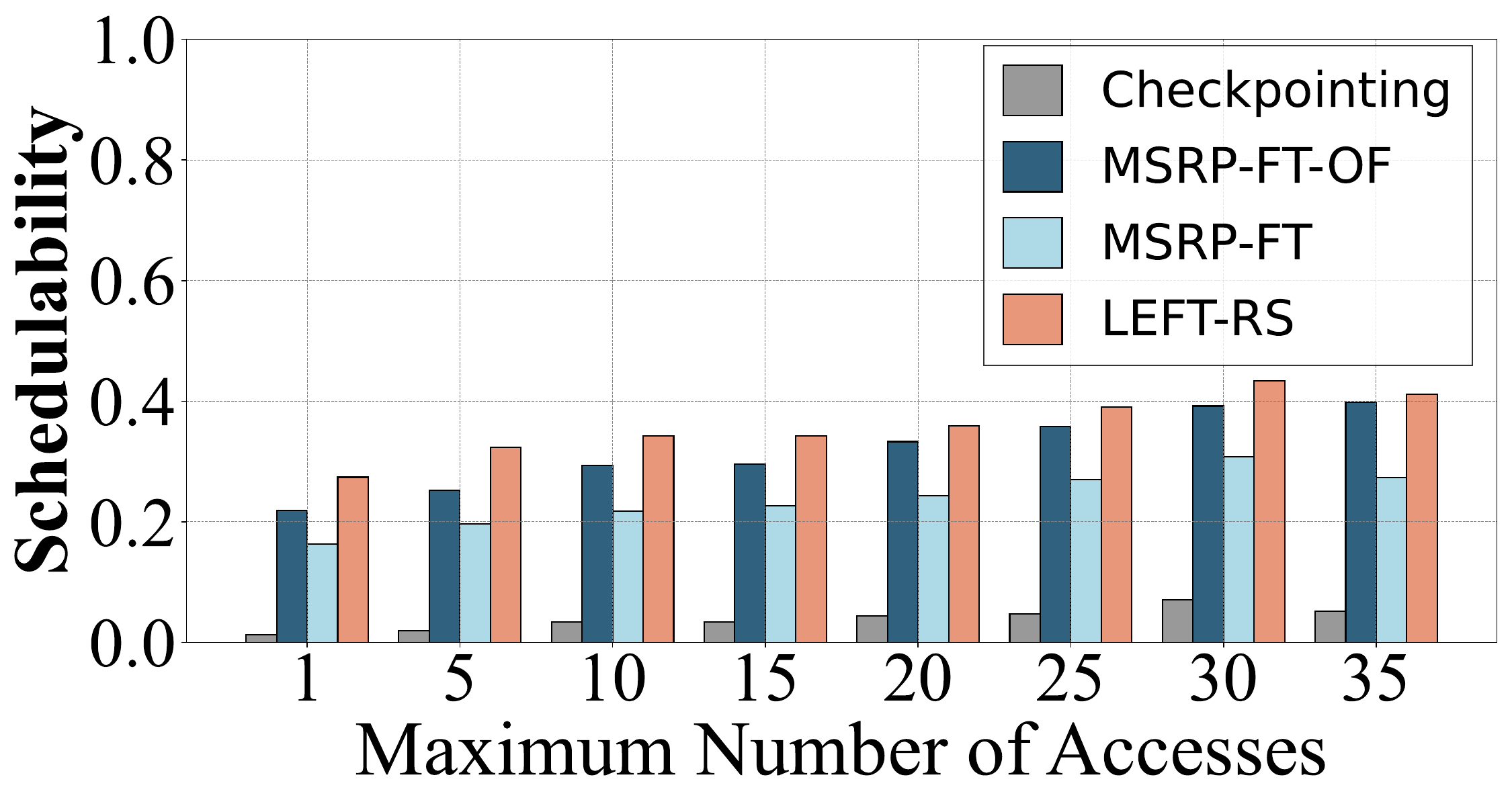}
        \caption{ Schedulability with varied $A$}
        \label{e3}
    \end{subfigure}
    \begin{subfigure}[b]{0.33\textwidth}
        \centering
        \includegraphics[width=\linewidth]{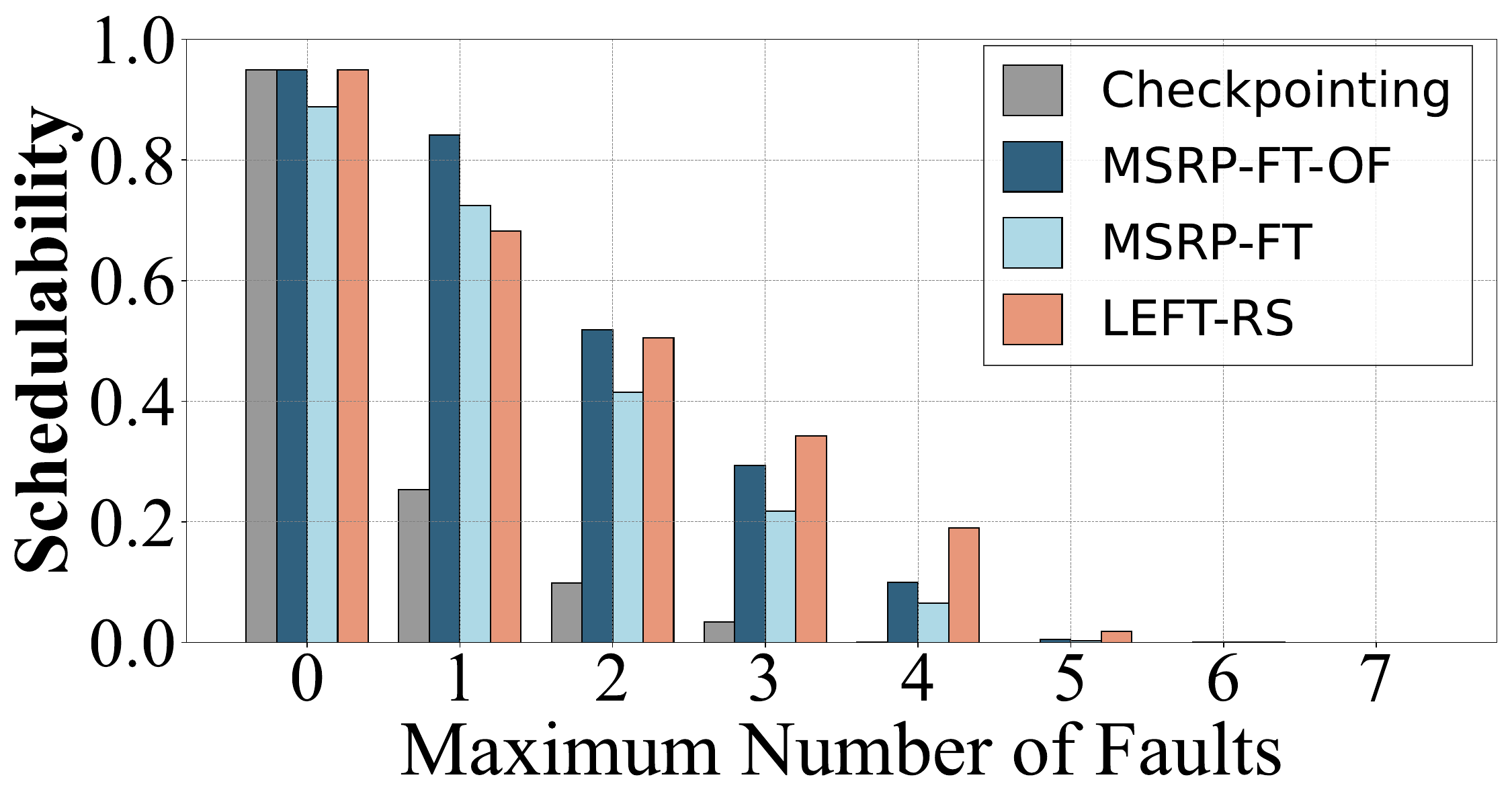}
        \caption{ Schedulability with varied $f$}
        \label{f3}
    \end{subfigure}

    \caption{System schedulability with $N=5$, $M=10$, $A=10$, $L =[1,100]$, $rsf=0.5$, $f=3$ and $K=M$ resources}
    \label{fig: schedulability}
    
\end{figure*}

In this section, we evaluate the effectiveness of the proposed approach by comparing the schedulability of systems under LEFT-RS, MSRP-FT, MSRP-FT-OF (Overhead Free), and Checkpointing. All methods adopt the same checkpoint placement strategy, ensuring a fair comparison without considering checkpointing overhead. Their difference lies solely in how they manage global resource sharing in the presence of faults.
MSRP-FT and MSRP-FT-OF both employ the helping mechanism introduced in Section~\ref{sec:MSRP-FT}. Specifically, MSRP-FT-OF does not account for protocol overhead, while MSRP-FT includes the overhead as defined in Equation~\eqref{eq:MSRP-FT_overhead}. By contrast, Checkpointing requires tasks to simply re-execute sequentially when a fault occurs while holding a lock.

We implemented the evaluations by generating synthetic tests. We consider multicore platforms with a set of cores $M=[2,4,6,8,10,12,14,16]$. For each core, we vary the number of tasks per core in the range of $N= [2,3,4,5,6,7,8,9]$, which enables a proper range of schedulable systems for analysis and comparison.
The utilisation of each task $U_{i}$ is generated using the UUnifast algorithm \cite{bini2005measuring}, with a total system utilisation bound $U_{\text{total}}=0.04 \cdot M\cdot N$. Task periods are randomly selected between [1ms,1000ms] in a log-uniform distribution. The deadlines for the tasks are equal to their periods. 
For a task $\tau_i$, the total worst-case execution (including critical sections) is given by $\widehat{C}_{i} = U_{i} \times T_i$.

Each system contains $K$ number of shared resources equal to the number of cores (i.e., $K=M$).
Among the generated tasks, a ratio of tasks are selected to request a random number of resources (up to $K$). The ratio is denoted by resource-sharing factor $rsf\in[0,0.8]$. The maximum number of accesses to a resource from each task is set to $A=[1,5,10,15,20,25,30,35] $ for a single release. The worst-case computation time of a shared resource is generated randomly in the range $L\in[1\mu s,300\mu s]$. We denote the total length of critical sections of $\tau_i$ as $C^{r}_{i}$, then we enforce
$C_i=\widehat{C}_{i}-C^{r}_{i}\geq 0$ to obtain  $C_i$. 
We set the maximum number of faults a task can incur during a single release to be $f \in [0,7]$. For example, if $f=5$, a task in one release can incur randomly from 0 up to 5 times faults. 
This setup follows the same fault model used in \cite{chen2022msrp} to ensure a fair comparison. It is well-justified, as radiation and EMI studies confirm that a single energetic particle or disturbance can induce multi-bit or multi-node upsets, making multiple transient faults per task a realistic and critical scenario in safety-critical embedded systems \cite{baumann2005radiation}.
In addition, tasks are allocated by the Worst-Fit heuristic and are scheduled by FP-FPS. Deadline Monotonic Priority Ordering is applied.

To parametrize the structural coordination overheads introduced by MSRP-FT, as discussed in Limitation~\ref{limit:complex}, we rely on empirical measurements from real-time operating systems deployed on multicore platforms. The wrap construction overhead, denoted as $O_{\text{wrap}}$ and $O_{\text{self\_wrap}}$ in Equation~\eqref{eq:MSRP-FT_overhead}, arises from the cost of publishing a shared descriptor that contains function pointers, resource references, and synchronisation flags. 
The underlying operations include memory setup, global visibility enforcement, and memory fencing. Shi et al.~\cite{shi2017implementation} report that enabling the helping mechanism in LITMUS\textsuperscript{RT} incurs a wrap setup cost consistently below \(1\,\mu s\). Similarly, a Linux-based fault-tolerant messaging prototype reports cross-core message publication latency of approximately \(0.55\,\mu s\)~\cite{losa2017transparent}. Based on this evidence, we conservatively assign $O_{\text{wrap}} = O_{\text{self\_wrap}} = 1\,\mu s$ in our evaluation.

The per-replica execution overhead, \(O_{\text{replica}}\) captures the cost incurred before a helper begins execution: polling for readiness, fetching the operation descriptor, and copying resource state. Zhao et al.~\cite{zhao2020complete} measured this cooperative execution latency under MrsP on LITMUS\textsuperscript{RT} at approximately \(8.4\,\mu s\). Shi et al.~\cite{shi2017implementation} further decomposed this into \(5.6\,\mu s\) for the migration itself and \(1.5\,\mu s\) for context re-entry on the helper core. Even commercial RTOS platforms show comparable orders of magnitude, e.g., \(11\,\mu s\) for VxWorks and \(13.4\,\mu s\) for RTLinux~\cite{ip2001performance}. Based on these results, we assign \(O_{\text{replica}} = 6\,\mu s\), which reflects a realistic and typical cost of helper-driven coordination and execution in multicore embedded systems, so as to avoid overstating the overhead in comparison with LEFT-RS.

In the following evaluations, for each combination of system settings, 1000 systems are generated, and the percentage of schedulable systems of the considered models is presented. 
In Figure \ref{fig: schedulability}, the y-axis is the rate of schedulable systems over 1000 generated systems (denoted as schedulability).

\textbf{Observation 1:} \textit{As shown in Figure \ref{fig: schedulability}, LEFT-RS consistently outperforms the Checkpointing approach across all settings.}

Since the two approaches differ only in their fault-tolerance mechanisms for global resources, this performance gap validates that the direct integration of the conventional fault-tolerance technique will lead to unmanageable global resource contention. Although checkpointing reduces re-execution costs by dividing tasks into segments, it does not mitigate contention in fault-tolerant global resource access. When faults occur, prolonged lock-holding from sequential re-executions can significantly delay other tasks, leading to system inefficiency. 

\textbf{Observation 2:} \textit{LEFT-RS approach consistently outperforms MSRP-FT in Figures \ref{a3}, \ref{b3}, \ref{c3}, \ref{d3} and \ref{e3}.}

As shown in Figure \ref{a3}, increasing the number of cores from 2 to 16 while keeping the per-core task count fixed at 
$N=5$ reduces schedulability for all methods. This decline is mainly due to heightened contention for global resources as more cores and tasks compete for access. 


The performance gap between LEFT-RS and MSRP-FT widens as the number of cores increases. For instance, with 4 cores, LEFT-RS  successfully schedules 10 more systems than MSRP-FT, whereas with 6 cores, this difference grows to 71 systems. This suggests that as global resource contention intensifies, LEFT-RS more effectively manages both resource access and fault tolerance. Although MSRP-FT employs a fault-tolerance mechanism by utilising spinning tasks to execute critical sections in parallel and accelerate recovery, it still has some fundamental limitations. As discussed in Limitations \ref{limit:lock} and \ref{limit:complex} (Section~\ref{sec:limitation}), tasks in MSRP-FT must manage resource accesses sequentially. If a task encounters a fault, it delays all subsequent tasks in the FIFO queue. Moreover, the unique overheads associated with MSRP-FT can be directly affected by the number of requests from different cores as illustrated in Equation \eqref{eq:MSRP-FT_overhead}.
In contrast, LEFT-RS allows tasks to execute their critical sections independently without following FIFO order, which not only enables them to exit the FIFO queue upon successful update but also exempts them from the overheads associated with MSRP-FT.
Overall, as shown in Figure~\ref{a3}, LEFT-RS outperforms MSRP-FT by an average of 51.3\%.

A similar trend is observed in Figures~\ref{b3}, \ref{c3}, and \ref{d3}, where variations in system parameters, such as the total number of tasks, the proportion of tasks that share resources, and the length of critical sections, lead to increased global resource contention. In all these scenarios, LEFT-RS consistently outperforms MSRP-FT. Specifically, it achieves average improvements of 62.9\%, 58.8\%, and 84.5\% in Figures~\ref{b3}, \ref{c3}, and \ref{d3}, respectively.
In addition, when $L=[1\mu s,15\mu s]$ in Figure \ref{d3}, the Checkpointing outperforms MSRP-FT, indicating that when the critical section is small, re-executing sequentially is more 
cost-effective than applying a helping mechanism.

As shown in Figure \ref{e3}, an increase in the maximum number of resource accesses does not necessarily degrade schedulability. This is because, under system utilisation constraints, an increase in accesses can sometimes shorten the critical section length of shared resources generated. Similar effects have been observed in \cite{zhao2024frap,chen2022msrp}. Nevertheless, LEFT-RS continues to outperform MSRP-FT by an average of 53\%. 

\textbf{Observation 3:} \textit{LEFT-RS loses its advantage over MSRP-FT when the number of faults in the system is very small. }

As shown in Figure~\ref{f3}, when $f=0$, Checkpointing performs exactly like traditional MSRP by definition, as no faults are considered in the system.
LEFT-RS performs identically to Checkpointing because the synchronisation overhead is avoided, according to Rule~\ref{rule10} in Section~\ref{sec:protocol}.
The additional overheads associated with MSRP-FT cause it to perform worse than the other protocols.
When the maximum number of faults per task is randomly set to at most 1, where most resource accesses remain fault-free, LEFT-RS loses some of its advantages to MSRP-FT by an average of 6.2\%. This behaviour matches the prediction in Corollary~\ref{cor:better_performance}, where each remote request does not hold the resource-holder position for long due to rare faults, and the overhead of MSRP-FT is manageable compared to the critical section length in this scenario.

However, as $f$ increases, the benefits of LEFT-RS become evident. In scenarios where LEFT-RS has an advantage (apart from $f=1$), it outperforms MSRP-FT by an average of 110\%.  
Under MSRP-FT, resource access follows a one-serve-at-a-time semantics, where the head-of-queue task occupies the exclusive service slot. If this task suffers frequent faults during its critical section, it can thereby extends the blocking time of subsequent tasks in the FIFO queue. In contrast, LEFT-RS limits the maximum resource access time units of each request to $n^x_i + m + 1$, as illustrated in Corollaries~\ref{cor:proposed} and \ref{cor:better_performance}, regardless of the number of faults a request encounters. This prevents excessive faults from degrading system performance.

\begin{table}[t]
\centering
\vspace{4pt}
\footnotesize
\caption{Number of Systems Schedulable Only by MSRP-FT (left) or Only by LEFT-RS (right)}
\resizebox{\linewidth}{!}{%
\begin{tabular}{|c|c|c|c|c|c|c|}
\hline
\multirow{2}{*}{$A$} & MSRP-FT~\cmark & LEFT-RS~\cmark & & \multirow{2}{*}{$f$} & MSRP-FT~\cmark & LEFT-RS~\cmark \\
                     & LEFT-RS~\xmark & MSRP-FT~\xmark & &                     & LEFT-RS~\xmark & MSRP-FT~\xmark \\
\hline
1  & 0   & 111 & & 0 & 0   & 61  \\
5  & 0   & 127 & & 1 & 52  & 10  \\
10 & 0   & 124 & & 2 & 0   & 90  \\
15 & 0   & 115 & & 3 & 0   & 124 \\
20 & 0   & 116 & & 4 & 0   & 125 \\
25 & 1   & 121 & & 5 & 0   & 15  \\
30 & 0   & 126 & & 6 & 0   & 0  \\
35 & 0   & 139 & & 7 & 0   & 0   \\
\hline
\end{tabular}%
}
\label{tab:schedulable_systems}
\vspace{-1em}
\end{table}

Table~\ref{tab:schedulable_systems} compares the number of systems that can be scheduled exclusively by either MSRP-FT or LEFT-RS.
On the left half of the table, the parameter \( A \) is varied, corresponding to the setting in Figure~\ref{e3}. The first column under this section (labelled “MSRP-FT~\checkmark / LEFT-RS~\xmark”) shows the number of systems that are schedulable only by MSRP-FT but not by LEFT-RS. The second column (labelled “LEFT-RS~\checkmark / MSRP-FT~\xmark”) represents the opposite.
Similarly, the right half of the table varies the parameter \( f \), as shown in Figure~\ref{f3}. 
This table highlights the complementary strengths of the two methods under different parameter configurations.

\textbf{Observation 4:} \textit{LEFT-RS not only outperforms MSRP-FT in terms of overall schedulability, but also in scenarios where MSRP-FT struggles.}

When varying \( A \), as shown in Table~\ref{tab:schedulable_systems}, LEFT-RS consistently demonstrates a one-way dominance over MSRP-FT. Specifically, it can schedule almost all systems deemed schedulable under MSRP-FT, and additionally handle a large number of systems that cannot be scheduled under MSRP-FT. When \( A = 25 \), there is only one system that is uniquely schedulable by MSRP-FT. This aligns with Corollary~\ref{cor:better_performance}, which states that we do not theoretically dominate MSRP-FT in every scenario. In this case, the setting of \( f = 3 \) (the maximum number of faults) allows tasks generated to encounter between 0 to 3 faults. The rare advantage for MSRP-FT arises when tasks experience few faults, and the associated overhead remains manageable relative to critical section lengths.

A similar observation can be made from the right part of the table when varying the fault tolerance level \( f \). When \( f = 1 \),  MSRP-FT successfully schedules more systems that cannot be scheduled by LEFT-RS. This is expected, as many tasks in these configurations are fault-free or experience only one fault, which limits the fault-tolerance design of LEFT-RS.
Even then, only 52 out of 1000 systems are missed by LEFT-RS. When \( f \) increases to 2, this number immediately drops to 0. As \( f \) increases further, the dominance of the LEFT-RS becomes clear.
These findings further reinforce the method’s reliability and effectiveness under increasingly fault-tolerant scenarios.

\textbf{Observation 5:} \textit{LEFT-RS maintains a clear advantage over MSRP-FT-OF.}

The comparison with MSRP-FT-OF is not practically meaningful, since MSRP-FT is inherently associated with external overheads absent in LEFT-RS. Nevertheless, the results remain informative. As shown in Figure \ref{fig: schedulability}, the overall schedulability trends follow the same pattern as with MSRP-FT: the relative differences between bars are preserved, but the performance gap compared to LEFT-RS is smaller. Although MSRP-FT-OF remains inferior to LEFT-RS in the majority of cases, at lower fault numbers (see Figure \ref{f3}) one additional instance ($f=2$) appears where MSRP-FT-OF achieves slightly higher schedulability than LEFT-RS. This behavior is consistent with Observation 3. Finally, the presentation of MSRP-FT-OF confirms that the advantage of LEFT-RS stems from its protocol design rather than solely from the exclusion of overheads.

\section{Conclusion}
\label{sec:conclusion}

In this paper, we propose a Lock-frEe Fault-Tolerant Resource Sharing (LEFT-RS) protocol for multicore real-time systems. Unlike MSRP-FT, where concurrency only serves the head request and access remains serialized, LEFT-RS enables true concurrency by allowing each request to execute locally after concurrent reading.
Tasks can complete earlier upon successful execution, significantly reducing blocking in the presence of frequent transient faults.
The independent execution avoids heavy coordination overhead
and enables fault resilience.
A comprehensive worst-case response time analysis is developed to support timing predictability. Evaluation results demonstrate that our method consistently outperforms state-of-the-art protocols, achieving up to an $84.5\%$ improvement in schedulability on average. 
Future extensions will move beyond non-preemptive execution and adopt more advanced scheduling strategies, such as ceiling-based protocols.

\section{Acknowledgment}
This research was funded in part by Innovate UK SCHEME project (10065634). EPSRC Research Data Management: No new primary data was created during this study. 

\bibliographystyle{IEEEtran}
\bibliography{ref}

@article{baumann2005radiation,
  title={Radiation-induced soft errors in advanced semiconductor technologies},
  author={Baumann, Robert C},
  journal={IEEE Transactions on Device and Materials Reliability},
  volume={5},
  number={3},
  pages={305--316},
  year={2005}
}

@inproceedings{izosimov2005design,
  title={Design optimization of time-and cost-constrained fault-tolerant distributed embedded systems},
  author={Izosimov, Viacheslav and Pop, Paul and Eles, Petru and Peng, Zebo},
  booktitle={IEEE Design, Automation and Test in Europe (DATE)},
  pages={864--869},
  year={2005}
}

@incollection{brandenburg2022multiprocessor,
  title={Multiprocessor real-time locking protocols},
  author={Brandenburg, Bj{\"o}rn B},
  booktitle={Handbook of Real-Time Computing},
  pages={347--446},
  year={2022},
  publisher={Springer}
}

@article{vijaykumar2002transient,
  title={Transient-fault recovery using simultaneous multithreading},
  author={Vijaykumar, TN and Pomeranz, Irith and Cheng, Karl},
  journal={ACM SIGARCH Computer Architecture News},
  volume={30},
  number={2},
  pages={87--98},
  year={2002}
}

@inproceedings{osinski2017survey,
  title={A survey of fault tolerance approaches on different architecture levels},
  author={Osinski, Lukas and Langer, Tobias and Mottok, Juergen},
  booktitle={30th International Conference on Architecture of Computing Systems (ARCS)},
  pages={1--9},
  year={2017},
}

@inproceedings{chen2022msrp,
  title={{MSRP-FT}: Reliable resource sharing on multiprocessor mixed-criticality systems},
  author={Chen, Nan and Zhao, Shuai and Gray, Ian and Burns, Alan and Ji, Siyuan and Chang, Wanli},
  booktitle={IEEE 28th Real-Time and Embedded Technology and Applications Symposium (RTAS)},
  pages={201--213},
  year={2022},
}

@inproceedings{zhao2017new,
  title={New schedulability analysis for {MrsP}},
  author={Zhao, Shuai and Garrido, Jorge and Burns, Alan and Wellings, Andy},
  booktitle={IEEE 23rd International Conference on Embedded and Real-time Computing Systems and Applications (RTCSA)},
  pages={1--10},
  year={2017}
}

@inproceedings{gai2003comparison,
  title={A comparison of {MPCP} and {MSRP} when sharing resources in the Janus multiple-processor on a chip platform},
  author={Gai, Paolo and Di Natale, Marco and Lipari, Giuseppe and Ferrari, Alberto and Gabellini, Claudio and Marceca, Paolo},
  booktitle={The 9th IEEE Real-Time and Embedded Technology and Applications Symposium (RTAS)},
  pages={189--198},
  year={2003}
}

@inproceedings{block2007flexible,
  title={A flexible real-time locking protocol for multiprocessors},
  author={Block, Aaron and Leontyev, Hennadiy and Brandenburg, Bjorn B and Anderson, James H},
  booktitle={13th IEEE International Conference on Embedded and Real-time Computing Systems and Applications (RTCSA)},
  pages={47--56},
  year={2007}
}

@article{brandenburg2013omlp,
  title={The {OMLP} family of optimal multiprocessor real-time locking protocols},
  author={Brandenburg, Bj{\"o}rn B and Anderson, James H},
  journal={Design Automation for Embedded Systems},
  volume={17},
  number={2},
  pages={277--342},
  year={2013},
  publisher={Springer}
}

@inproceedings{zhao2024frap,
  title={{FRAP}: A Flexible Resource Accessing Protocol for Multiprocessor Real-Time Systems},
  author={Zhao, Shuai and Xu, Hanzhi and Chen, Nan and Su, Ruoxian and Chang, Wanli},
  booktitle={IEEE Real-Time Systems Symposium (RTSS)},
  pages={349--361},
  year={2024}
}

@article{anderson1997real,
  title={Real-time computing with lock-free shared objects},
  author={Anderson, James H and Ramamurthy, Srikanth and Jeffay, Kevin},
  journal={ACM Transactions on Computer Systems (TOCS)},
  volume={15},
  number={2},
  pages={134--165},
  year={1997}
}

@article{nabavi2023fault,
  title={A fault-tolerant resource locking protocol for multiprocessor real-time systems},
  author={Nabavi, Seyede Sahebeh and Farbeh, Hamed},
  journal={Microelectronics Journal},
  volume={137},
  pages={105809},
  year={2023},
  publisher={Elsevier}
}

@inproceedings{gai2001minimizing,
  title={Minimizing memory utilization of real-time task sets in single and multi-processor systems-on-a-chip},
  author={Gai, Paolo and Lipari, Giuseppe and Di Natale, Marco},
  booktitle={22nd IEEE Real-Time Systems Symposium (RTSS)},
  pages={73--83},
  year={2001}
}

@inproceedings{shi2017implementation,
  title={Implementation and evaluation of multiprocessor resource synchronization protocol ({MrsP}) on {LITMUS^{RT}}},
  author={Shi, Junjie and Chen, Kuan-Hsun and Zhao, Shuai and Huang, Wen-Hung and Chen, Jian-Jia and Wellings, Andy},
  booktitle={13th Workshop on Operating Systems Platforms for Embedded Real-Time Applications},
  year={2017}
}

@article{zhao2020complete,
  title={A complete run-time overhead-aware schedulability analysis for {MrsP} under nested resources},
  author={Zhao, Shuai and Garrido, Jorge and Wei, Ran and Burns, Alan and Wellings, Andy and de la Puente, Juan A},
  journal={Journal of Systems and Software},
  volume={159},
  pages={110449},
  year={2020},
  publisher={Elsevier}
}

@inproceedings{losa2017transparent,
  title={Transparent fault-tolerance using intra-machine full-software-stack replication on commodity multicore hardware},
  author={Losa, Giuliano and Barbalace, Antonio and Wen, Yuzhong and Chuang, Ho-Ren and Ravindran, Binoy and Sadini, Marina},
  booktitle={IEEE 37th International Conference on Distributed Computing Systems (ICDCS)},
  pages={1521--1531},
  year={2017}
}

@inproceedings{burns2013schedulability,
  title={A schedulability compatible multiprocessor resource sharing protocol--{MrsP}},
  author={Burns, Alan and Wellings, Andy J},
  booktitle={IEEE 25th Euromicro Conference on Real-Time Systems (ECRTS)},
  pages={282--291},
  year={2013}
}

@article{punnekkat2001analysis,
  title={Analysis of checkpointing for real-time systems},
  author={Punnekkat, Sasikumar and Burns, Alan and Davis, Robert},
  journal={Real-Time Systems},
  volume={20},
  pages={83--102},
  year={2001},
  publisher={Springer}
}

@article{bini2005measuring,
  title={Measuring the performance of schedulability tests},
  author={Bini, Enrico and Buttazzo, Giorgio C},
  journal={Real-Time Systems},
  volume={30},
  number={1},
  pages={129--154},
  year={2005},
  publisher={Springer}
}

@article{ip2001performance,
  title={Performance analysis of {VxWorks} and {RTLinux}},
  author={Ip, Benjamin},
  journal={Languages of Embedded Systems Department of Computer Science},
  year={2001}
}

@inproceedings{miskov2008process,
  title={Process variability-aware transient fault modelling and analysis},
  author={Miskov-Zivanov, Natasa and Wu, Kai-Chiang and Marculescu, Diana},
  booktitle={IEEE/ACM International Conference on Computer-Aided Design},
  pages={685--690},
  year={2008}
}

\end{document}